\documentclass[a4paper,11pt]{article}
\usepackage[a4paper,top=3.5cm,bottom=3cm,left=3cm,right=3cm]{geometry}

\usepackage[sectionbib]{natbib}
\usepackage{setspace}

\usepackage[utf8]{inputenc} % allow utf-8 input
\usepackage[T1]{fontenc}    % use 8-bit T1 fonts
\usepackage{hyperref}       % hyperlinks
\usepackage{url}            % simple URL typesetting
\usepackage{booktabs}       % professional-quality tables
\usepackage{amsfonts,amsmath}       % blackboard math symbols
\usepackage{amsthm}
\usepackage{nicefrac}       % compact symbols for 1/2, etc.
\usepackage{microtype}      % microtypography
\usepackage{algorithm,algorithmic}
\usepackage{graphicx}
\usepackage{pdflscape}

\usepackage{latexsym}

\usepackage{color}

\newcommand{\indep}{\perp \!\!\! \perp}

\newtheorem{theorem}{Theorem}
\newtheorem{corollary}{Corollary}
\newtheorem{assumption}{Assumption}
\newtheorem*{assumption1p}{Assumption 1'}
\newtheorem{lemma}{Lemma}

\usepackage{lineno,hyperref}
\newcommand*\patchAmsMathEnvironmentForLineno[1]{
  \expandafter\let\csname old#1\expandafter\endcsname\csname #1\endcsname
  \expandafter\let\csname oldend#1\expandafter\endcsname\csname end#1\endcsname
  \renewenvironment{#1}
     {\linenomath\csname old#1\endcsname}
     {\csname oldend#1\endcsname\endlinenomath}}
\newcommand*\patchBothAmsMathEnvironmentsForLineno[1]{
  \patchAmsMathEnvironmentForLineno{#1}
  \patchAmsMathEnvironmentForLineno{#1*}}
\AtBeginDocument{
\patchBothAmsMathEnvironmentsForLineno{equation}
\patchBothAmsMathEnvironmentsForLineno{align}
\patchBothAmsMathEnvironmentsForLineno{flalign}
\patchBothAmsMathEnvironmentsForLineno{alignat}
\patchBothAmsMathEnvironmentsForLineno{gather}
\patchBothAmsMathEnvironmentsForLineno{multline}
}
\modulolinenumbers[5]

\title{Outlier-Resistant Estimators for Average Treatment Effect in Causal Inference}
\author{
    Kazuharu Harada \thanks{Tokyo Medical University, Japan. e-mail: \texttt{haradak@tokyo-med.ac.jp}}
        \and
    Hironori Fujisawa
    \thanks{The Institute of Statistical Mathematics, Japan. e-mail: \texttt{fujisawa@ism.ac.jp}}
    }
\date{\today}

\begin{document}
\maketitle

\begin{abstract}
The inverse probability (IPW) and doubly robust (DR) estimators are often used to estimate the average causal effect (ATE), but are vulnerable to outliers. The IPW/DR median can be used for outlier-resistant estimation of the ATE, but the outlier resistance of the median is limited and it is not resistant enough for heavy contamination. We propose extensions of the IPW/DR estimators with density power weighting, which can eliminate the influence of outliers almost completely. The outlier resistance of the proposed estimators is evaluated through the unbiasedness of the estimating equations. Unlike the median-based methods, our estimators are resistant to outliers even under heavy contamination. Interestingly, the naive extension of the DR estimator requires bias correction to keep the double robustness even under the most tractable form of contamination. In addition, the proposed estimators are found to be highly resistant to outliers in more difficult settings where the contamination ratio depends on the covariates. The outlier resistance of our estimators from the viewpoint of the influence function is also favorable. Our theoretical results are verified via Monte Carlo simulations and real data analysis. The proposed methods were found to have more outlier resistance than the median-based methods and estimated the potential mean with a smaller error than the median-based methods.
\end{abstract}

% keywords can be removed
% \keywords{Robust Statistics \and Causal Inference \and Missing Data Analysis \and M-estimator}

\section{Introduction}
% \linenumbers

Statistical causal inference provides various estimators for causal quantities like the average treatment effect (ATE). To estimate such quantities, the propensity score is widely applied in various ways, such as stratification, matching, inverse probability weighting (IPW), and the doubly robust (DR) estimator \citep*{robins1994estimation,rosenbaum1983central,bang2005doubly}. These estimators are designed to control confounding, and they are consistent with the target quantity under some assumptions. 

As discussed in the later section, the IPW and DR estimators are vulnerable to outliers since they partially use the sample mean. An outlier in a multivariate setting is classified into three types: a vertical outlier, a good leverage point, and a bad leverage point \citep*{Rousseeuw1990-wz}. Figure \ref{fig:threetypes} illustrates the three types of outliers. \cite{Canavire-Bacarreza2021-dc} has investigated the influence of these types of outliers on the estimators of the ATE including IPW through exhaustive Monte-Carlo simulations; they have pointed out that the vertical outliers in the outcome variable lead to a serious bias of the ATE estimation. In this paper, we are interested in reducing the bias caused by the vertical outliers.
\begin{figure}[htbp]
    \centering
    \includegraphics[width=0.6\textwidth]{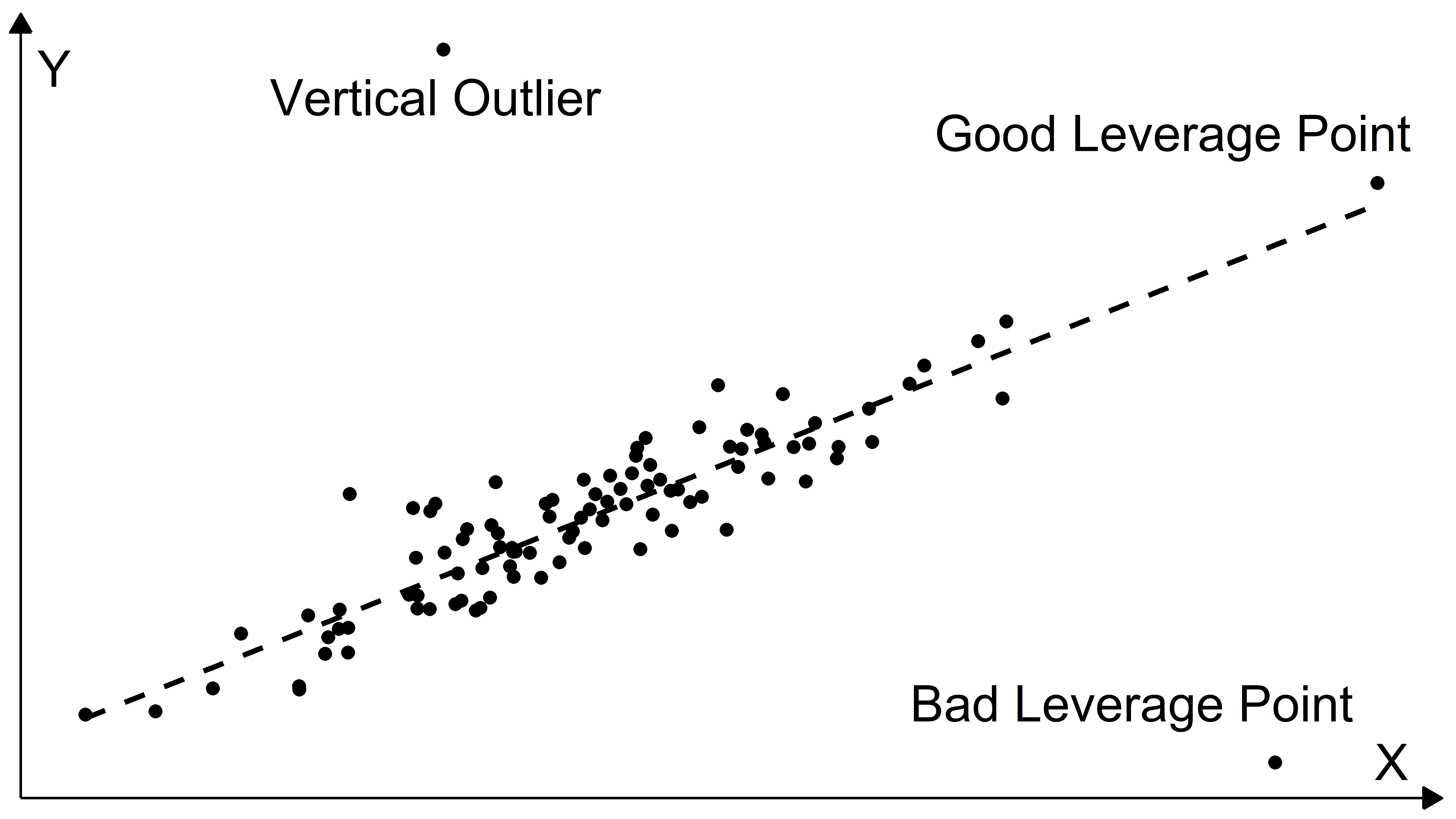}
    \caption{Three types of outliers.}
    \label{fig:threetypes}
\end{figure}
\vspace{-4mm}

Outlier-resistant statistics have been studied for long; however, most literature does not consider a causal setting \citep{huber2004robust,hampel2011robust,maronna2019robust}.
The established methods of outlier-resistant statistics are not directly applicable to causal settings.
The median-based estimators are the only examples which are applicable to the estimation of the ATE under outlier contamination (\citeauthor{firpo2007efficient}, \citeyear{firpo2007efficient}; \citeauthor{zhang2012causal}, \citeyear{zhang2012causal}; \citeauthor{diaz2017efficient}, \citeyear{diaz2017efficient}; \citeauthor*{sued2020robust}, \citeyear{sued2020robust}).
It is well known that the sample median is more resistant to outliers than the sample mean, but it is still affected; in particular, when the contamination ratio is not small and the outliers lie on one side of the data-generating density, the influence becomes so large as it cannot be ignored \citep{fujisawa2008robust}.

In this paper, we propose extensions of the IPW and DR estimators for the mean of the potential outcome with more outlier resistance than the median-based methods. We discuss the outlier resistance of these estimators from the viewpoint of the unbiasedness of the estimating equation and influence function (IF). 
In most literature on outlier-resistant statistics, the contamination ratio is assumed to be small and be independent of covariates; however, we discuss the outlier resistance of the proposed estimators under more general assumptions, including the case where the contamination ratio is not small and related to covariates. Interestingly, a straight extension of the DR estimator loses the robustness to model misspecification under contamination. We also propose a bias-corrected version of the extended DR estimator, which holds the double robustness under contamination. Furthermore, the theoretical advantages of our estimators are verified through Monte-Carlo simulations and real data analysis.

The remainder of this paper is organized as follows. In Section \ref{sec:prelim}, we introduce the potential outcome framework for causal inference and the basic concept of outliers. In Section \ref{sec:ORsemi}, we propose novel estimators and discuss the outlier resistance from the viewpoint of the unbiasedness of the estimating equations. In Section \ref{sec:if}, we evaluate the outlier resistance in terms of the IF. In Section \ref{sec:asymp}, we discuss asymptotic properties. Finally, in Sections 6 and 7, we present the experimental results.

%%%%%%%%%%%%%%%%%%%%%%%%%%%%%%%%%%%%%%%%%%%%%%%%%%%%%%%%%%%%%%%%%%%%%%%%%%%%%%%%%%%%%%%%%%%%%%%%%%%%%%%%%%%%%%%%%%%%%%%

%%%%%%%%%%%%%%%%%%%%%%%%%%%%%%%%%%%%%%%%%%%%%%%%%%%%%%%%%%%%%%%%%%%%%%%%%%%%%%%%%%%%%%%%%%%%%%%%%%%%%%%%%%%%%%%%%%%%%%%%%%%%
\section{Preliminaries}\label{sec:prelim}
\subsection{Potential Outcome and Treatment Effect}\label{sec:ATE}
Let $(Y,T,X)$ be the observable random variables; $X$ is the outcome, $T$ is the treatment, and $X$ is the confounder. We assume that $Y$ is continuous and $T$ is binary; it is easy to extend $T$ to multiple discrete treatments. We have the observations $(Y_i,T_i,X_i)_{i=1}^n$ drawn from the distribution of $(Y,T,X)$ in an i.i.d. manner. 
Denote the potential outcome under $T=t$ by $Y^{(t)}$ and let $\mu^{(t)}=\mathbb{E}[Y^{(t)}]$. $Y^{(t)}$ is uniquely defined for every treatment as a random variable, namely, well-defined. Note that i.i.d. sampling and well-definedness of the potential outcome are collectively called the stable unit treatment value assumption \citep[SUTVA;][]{Imbens2015-py}. The ATE is defined as $\mu^{(1)}-\mu^{(0)}$. The ATE cannot be estimated directly since we cannot observe $Y^{(1)}$ and $Y^{(0)}$ simultaneously; instead, we use the observed variables under the common assumptions \citep[e.g.][]{Imbens2015-py}:
\begin{enumerate}
    \item \textit{Conditional Unconfoundedness}: $Y^{(t)} \indep T | X$ for all $t\in\{0,1\}$,
    \item \textit{Consistency}: $Y=Y^{(t)}$ if $T=t$,
    \item \textit{Positivity}: $P(T=1|X)>c$ for some constant $c>0$.
\end{enumerate}
The ATE can be estimated from the observed variables under these assumptions, namely, identifiable. Hereafter, we assume the triple assumption holds and focus on the estimation of $\mu^{(1)}$ for simplicity. $\mu^{(0)}$ is estimated in a similar way, and the ATE is estimated by the difference between the estimates of $\mu^{(1)}$ and $\mu^{(0)}$.

We introduce three consistent estimators of the potential mean.
The IPW estimator \citep{rosenbaum1983central} is based on the propensity score (PS). Let $\pi(x;\alpha)\in(0,1)$ be the PS, which models $P(T=1|x)$. We assume the PS is correctly specified, in other words, there exists $\alpha^*$ such that $\pi(x;\alpha^*)=P(T=1|x)$ for every $x$. The IPW estimator has several forms \citep{lunceford2004stratification}, but we use the weighted average form: $\hat{\mu}^{(1)}_{IPW} = \left.\left(\sum_{i=1}^nT_iY_i/\pi(X_i;\hat{\alpha})\right)\right/\left(\sum_{i=1}^nT_i/\pi(X_i;\hat{\alpha})\right)$,
where $\hat{\alpha}$ is an estimate of $\alpha$ obtained in a consistent manner such as the maximum likelihood estimation (MLE). The IPW estimator is consistent with $\mu^{(1)}$ if the PS model is correctly specified. The IPW estimator can be seen as the root of the following estimating equation:
\begin{align}
    \sum_{i=1}^n\frac{T_i}{\pi(X_i;\hat{\alpha})}(Y_i-\mu)=0\label{eq:esteqIPW}.
\end{align}
Outcome regression (OR) is also popular. To construct the OR estimator, we model $\mathbb{E}[Y|T=1,X]$ by some function $m_1(X;\beta)$. Then, the OR estimator is obtained as $n^{-1}\sum_{i=1}^{n}m_1(X_i;\hat{\beta})$, where $\hat{\beta}$ is a consistent estimate of $\beta$. The IPW and OR estimators are asymptotically consistent with $\mu^{(1)}$ when the model used in each estimator is correctly specified; the consistency is not assured if the model is misspecified.
The DR estimator \citep*{scharfstein1999adjusting,bang2005doubly} combines the IPW and OR estimators. Since the DR estimator is consistent with $\mu^{(1)}$ if either the PS or OR model is correctly specified, it is said to be ``doubly robust.'' 
Besides, if both models are correctly specified, the DR estimator is semiparametrically efficient \citep{robins1995semiparametric,Tsiatis2006-sz}.
Although there are many estimators equipped with double robustness, we refer the root of the following estimating equation as the DR estimator $\hat{\mu}^{(1)}_{DR}$, which is a special case of the augmented IPW estimator:
\begin{align}
    \sum_{i=1}^n\left[\frac{T_i}{\pi(X_i;\hat{\alpha})}(Y_i-\mu) - \frac{T_i - \pi(X_i;\hat{\alpha})}{\pi(X_i;\hat{\alpha})}\{m_1(X_i;\hat{\beta})-\mu\}\right]=0.\label{eq:esteqDR}
\end{align}

\subsection{IPW/DR M-estimators}
Let $\sum_{i=1}^n\psi(Y_i,\theta)=0$ be an estimating equation, where $\psi$ is a known vector-valued map, and $\theta$ is the parameter of interest. An estimator $\hat{\theta}$ solving the estimating equation is called an M-estimator. M-estimator is a large class of estimators involving MLE, IPW, OR, and DR. If the estimating equation is unbiased, say $\mathbb{E}_{\theta}[\psi(Y,\theta)]=0$, the M-estimator is consistent with the truth under some regularity conditions \citep[e.g. Chap. 5 of][]{van2000asymptotic}.

By replacing $Y_i-\mu$ in \eqref{eq:esteqDR} with an estimating function $\psi(Y_i;\theta)$, the IPW and DR estimators can be expanded to a general M-estimator.
If we are interested in the same parameter $\theta$ with respect to $Y^{(1)}$, the IPW and DR M-estimators \citep{Tsiatis2006-sz} are available:
\begin{gather}
    \sum_{i=1}^n\frac{T_i}{\pi(X_i;\hat{\alpha})}\psi(Y_i;\theta)=0\label{eq:esteqIPWM},\\
    \sum_{i=1}^n\left[\frac{T_i}{\pi(X_i;\hat{\alpha})}\psi(Y_i;\theta) - \frac{T_i - \pi(X_i;\hat{\alpha})}{\pi(X_i;\hat{\alpha})}\mathbb{E}_{\hat{q}}[\psi(Y_i;\theta)|T=1,X_i]\right]=0.\label{eq:esteqDRM}
\end{gather}
The conditional expectation $\mathbb{E}_{\hat{q}}[\psi(Y_i;\theta)|T=1,X_i]$ is calculated using the parametric OR model $q(y|T=1,x;\hat{\beta})$ via direct calculation or Monte-Carlo approximation \citep{hoshino2007doubly}. When the original M-estimating equation is unbiased, the IPW/DR estimating equations are unbiased under the proper model specification. 
The asymptotic properties of the IPW and DR M-estimators follow from the standard theory of M-estimators.

\subsection{Outlier-resistant Estimation}\label{sec:Mest}
This section provides a brief review of the outlier-resistant estimation of the mean in a one-variable and non-causal setting.
Let $\tilde{g}$ be the density function of a random variable $Z\in\mathbb{R}$. Assume that the density is contaminated as $\tilde{g}(z)=(1-\varepsilon)f_{\mu^*}(z)+\varepsilon\delta(z)$, where $f_{\mu^*}$ is the density of $Z$ without contamination equipped with the mean $\mu^*$, $\varepsilon$ is the contamination ratio, and $\delta$ is the density of outliers. 
Our goal is to estimate $\mu^*$ from i.i.d. observations $\{Z_1,...,Z_n\}$ drawn from $\tilde{g}$. If we model the contamination in this way, the sample mean converges to $(1-\varepsilon)\mu^{*} + \varepsilon\mathbb{E}_\delta[Z]$; if the mean of outliers is far from $\mu^*$, the sample mean is asymptotically biased. 
To deal with contamination, many types of M-estimators are applied. The unbiasedness of the estimating equation does not usually hold under contamination because
\begin{align}
    \mathbb{E}_{\tilde{g}}[\psi(Z,\mu^{*})] = (1-\varepsilon)\underbrace{\mathbb{E}_{f_{\mu^*}}[\psi(Z,\mu^{*})]}_{=0} + \varepsilon\mathbb{E}_\delta[\psi(Z,\mu^{*})] \ne 0.
\end{align} 
By designing $\psi$ to eliminate or bound $\mathbb{E}_\delta[\psi(Z,\mu^{*})]$, the influence of outliers can be reduced. Let $\theta_{\psi}^*$ denote a root of $\mathbb{E}_{\tilde{g}}[\psi(Z,\theta)]= 0$; then, the latent bias is defined as $\theta_{\psi}^* - \theta^*$. If $\delta$ is Dirac's delta and $\varepsilon$ is sufficiently small, the latent bias is approximated by the IF. The IF-based discussion in Section \ref{sec:if} provides some insights into the outlier resistance of the estimators when the contamination ratio is small. The latent bias and M-estimators are discussed in detail elsewhere \citep[e.g.][]{huber2004robust,fujisawa2013normalized,fujisawa2008robust}.

\subsection{IPW and DR Under Contamination}\label{sec:lbias}
Next, we move to a causal setting. In this paper, we consider the vertical outliers. In other words, we assume that only the outcome $Y$ may be contaminated, and that the contamination does not affect the causal mechanism among $(Y,T,X)$. A typical example is contamination of laboratory values in medical research with foreign substances.
Let $\delta_{Y|TX}$ be the conditional density of outliers given $(T,X)$, and let $\varepsilon_t(x)$ be the contamination ratio. Then, the contaminated conditional density given $(T,X)$ is defined as
\begin{align}\label{eq:contam1}
    \tilde{g}_{Y|TX}(y|t,x) = (1-\varepsilon_t(x))g_{Y|TX}(y|t,x) + \varepsilon_t(x)\delta_{Y|TX}(y|t,x),
\end{align}
where $g$ without tilde denotes the density without contamination; the tilde indicates that the distribution is contaminated. For simplifying the notations, we often drop the subscripts of density functions as long as there would be no confusion and write $\delta_{t}(y|x)=\delta_{Y|TX}(y|t,x)$ below.
The contamination ratio and their density depend on the treatment $T$ and the confounder $X$. Since we estimate $\mu^{(t)}$ for each treatment separately, the dependence on $T$ is tractable. In contrast, the dependence on $X$ is critical in our analysis. The $X$-dependent contamination is referred to as heterogeneous contamination. We also discuss the special case in which $\varepsilon$ and $\delta$ are not dependent on $X$, called homogeneous contamination. Note that we do not assume $\varepsilon_t(x)$ to be small enough to be negligible, except in Section \ref{sec:if}.

We are interested in the marginal mean of $Y^{(1)}$.
Let $f_{Y^{(1)}}(y;\mu^{(1)})$ be the true marginal density of $Y^{(1)}$. It is obtained by integrating $X$ out from $g_{Y|TX}(y|T,X)$ under $T=1$:
\begin{align}\label{eq:contam2}
    f_{Y^{(1)}}(y;\mu^{(1)}) 
        = \int g_{Y^{(1)}|X}(y|x)g_X(x)dx
        = \int g_{Y|TX}(y|1,x)g_X(x)dx.
\end{align}
The second equality holds from the triple assumption in Section \ref{sec:ATE}. We often write $f_{Y^{(1)}}(y;\mu^{(1)})$ as $f_1(y)$ for simplicity.

Under contamination, the IPW estimating equation is severely biased even if the true PS is obtained as $\pi(X|\alpha^*)=P(T=1|X)$:
\begin{align}\label{eq:IPWbias}
    \mathbb{E}_{\tilde{g}}\left[\frac{T}{\pi(X|\alpha^*)}(Y-\mu^{(1)})\right] 
    = \mathbb{E}_{g}\left[
        \varepsilon_1(X)\mathbb{E}_{-g+\delta}\left[(Y-\mu^{(1)})|X\right]
    \right] \ne 0.
\end{align}
It is found that the remaining term contains the expectation of $Y$ with respect to  $\delta$, which implies the estimating equation is severely affected by outliers. The DR estimating equation is also biased. To estimate $\mu^{(1)}$ accurately, we have to remove the influence of contamination.

%%%%%%%%%%%%%%%%%%%%%%%%%%%%%%%%%%%%%%%%%%%%%%%%%%%%%%%%%%%%%%%%%%%%%%%%%%%%%%%%%%%%%%%%%%%%%%%%%%%%%%%%%%%%%%%%%%%%%%%%%%%%

\section{Outlier-Resistant Extensions of IPW and DR}\label{sec:ORsemi}
Before we propose novel estimators, we introduce an assumption on outliers. Intuitively, we assume that the outliers are sufficiently far from the main outcome density. Figure \ref{fig:assumption1} shows real examples of outliers that satisfy this assumption. It is found that these outliers are far from the main body of the density both conditionally and marginally. 
\begin{figure}[htbp]
    \centering
    \includegraphics[width=\textwidth]{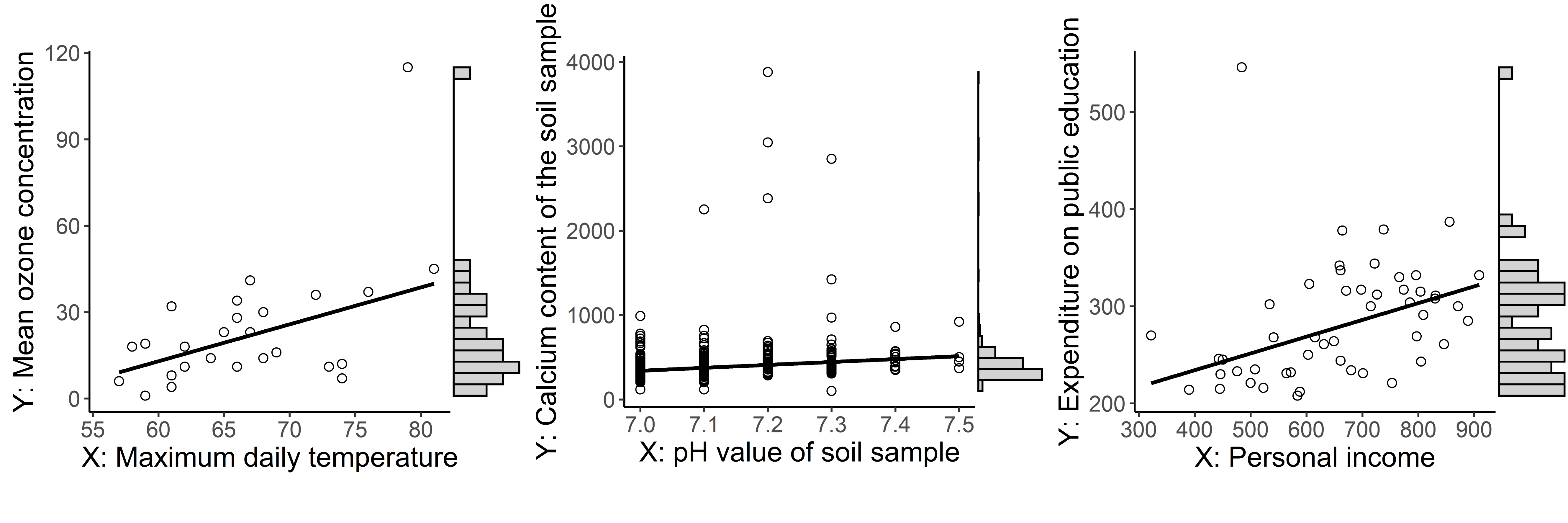}
    \caption{Real examples of outliers which satisfy Assumption 1. All datasets are included in the R package ``robustbase'' \citep{maechler2021package}: airmay (left), condroz (center), education (right).}
    \label{fig:assumption1}
\end{figure}

To formulate this assumption, we introduce density power weighting.
The density power weight is used to enhance the outlier resistance in non-causal settings \citep{windham1995robustifying,basu1998robust,jones2001comparison,fujisawa2008robust}.
Let $h(y;\mu)^{\gamma}~(\gamma>0)$ be a density power weight for $Y^{(1)}$, where $h(y;\mu)$ is a symmetric density function with the location parameter $\mu$. The density $h(y;\mu^{(1)})$ is not necessarily equal to the true marginal density $f_1(y)$, but we assume that both $h$ and the true density $f_1(y)$ are symmetric about $\mu^{(1)}$. The assumption of symmetry is common in outlier-resistant estimation, and it is a prerequisite to use the sample median as an estimator of the population mean. Any symmetric density is available for $h$ as long as it satisfies Assumption \ref{as:Outlier} below. Typically, we assume $h$ is Gaussian.
The tuning parameter $\gamma$ controls the variability of the weight; this leads to the trade-off between outlier resistance and asymptotic efficiency. Assumption 1 formally describes the assumption on outliers.
 
\begin{assumption}\label{as:Outlier}
Let $h(y;\mu)$ be a weighting density symmetric about $\mu$. Then, there exists $\gamma>0$ such that 
\begin{align}\label{eq:assumOutlier}
    \xi_1(X,\gamma) = \int \delta_1(y|X)h(y;\mu^{(1)})^{\gamma}(y-\mu^{(1)})dy \approx 0~~~~a.e.
\end{align}
\end{assumption}

\noindent Denote an arbitrary bounded function by $\phi(x)$. Assumption \ref{as:Outlier} implies 
\begin{align}
    \nu_1(\phi) := \mathbb{E}[\phi(X)\xi_1(X,\gamma)] = \int \phi(x)\xi_1(x,\gamma)g_X(x)dx \approx 0.
\end{align}
In particular, let $\phi(x)=1$; then, the outliers are marginally negligible: 
\begin{align}
    \nu_1(1) = \mathbb{E}[\xi_1(X,\gamma)] = \int \delta_1(y)h(y;\mu^{(1)})^{\gamma}(y-\mu^{(1)})dy \approx 0.
\end{align}
Throughout this paper, we assume that $\gamma$ is sufficiently large so that Assumption \ref{as:Outlier} holds. Assumption \ref{as:Outlier} is reduced to a simpler form when $\delta_1(y|X)$ is Dirac's delta at $y_0$; this is one of the core assumptions in Section \ref{sec:if}.
\begin{assumption1p}\label{as:outlier1p}
Let $h(y;\mu)$ be a weighting density that is symmetric about $\mu$, and assume that the density of outliers is Dirac's delta at $y_0~(\ne \mu^{(1)})$, say $\delta_{y_0}(y)$. Then, there exists $\gamma>0$ such that 
\begin{align}\label{eq:outlier1p}
    \int \delta_{y_0}(y)h(y;\mu^{(1)})^{\gamma}(y-\mu^{(1)})dy 
        = h(y_0;\mu^{(1)})^{\gamma}(y_0-\mu^{(1)}) \approx 0.
\end{align}
\end{assumption1p}
\noindent For example, if $h$ is a Gaussian density with mean $\mu^{(1)}$ and fixed variance, \eqref{eq:outlier1p} tends to 0 as $|y_0| \rightarrow \infty$ for fixed $\gamma>0$ since $h(y_0;\mu^{(1)})^{\gamma}(y_0-\mu^{(1)}) \propto \exp{(-\gamma (y_0-\mu^{(1)})^2)}(y_0-\mu^{(1)})$.

% \noindent For example, if $h$ is a Gaussian density with mean $\mu^{(1)}$ and fixed variance, \eqref{eq:outlier1p} tends to 0 as $\gamma \rightarrow \infty$ for fixed $\gamma$at every point $y_0\ne\mu^{(1)}$ since $h(y_0;\mu^{(1)})^{\gamma}\propto \exp{(-\gamma (y_0-\mu^{(1)})^2)}$; then, Assumption 1a holds.

\subsection{IPW-type Estimator}\label{sec:DP-IPW}
First, we introduce an extension of the IPW estimator, called the density power inverse probability weighting (DP-IPW) estimator.
The DP-IPW estimator is defined as a root of the following estimating equation:
\begin{align}\label{eq:DPIPWeq}
    \sum_{i=1}^n \frac{T_i}{\pi(X_i;\hat{\alpha})}h(Y_i;\mu)^{\gamma}(Y_i-\mu) = 0.
\end{align} 
Under no contamination, the DP-IPW estimating equation is unbiased. 

\begin{theorem}\label{thm:dpipw1}
Assume that the true propensity score $\pi(X;\alpha^*)$ is given. Then, under no contamination, we have 
\begin{align}
    \mathbb{E}_{g}\left[
        \frac{T}{\pi(X;\alpha^*)}h(Y;\mu^{(1)})^{\gamma}(Y-\mu^{(1)})
    \right] = 0.
\end{align} 
\end{theorem}
\noindent Only an estimate $\pi(X;\hat{\alpha})$ is available in practice, but the asymptotic consistency of (DP-)IPW still holds if the model $\pi(X;\alpha)$ is correctly specified.

Now we consider the contaminated case. The bias of the DP-IPW estimating equation takes a different form from \eqref{eq:IPWbias}.

\begin{theorem}\label{thm:dpipw2}
Assume $Y$ is contaminated as \eqref{eq:contam1}. 
Under the same assumptions as those in Theorem \ref{thm:dpipw1}, the expectation of the DP-IPW estimating equation is expressed as 
\begin{gather}\label{eq:biasthm2}
    \mathbb{E}_{\tilde{g}}\left[
        \frac{T}{\pi(X;\alpha^*)}h(Y;\mu^{(1)})^{\gamma}(Y-\mu^{(1)})
    \right] = B_1 +  \nu_1(\varepsilon_1), \\
    \text{where}~~~ B_1 =
    -\int \varepsilon_1(x)\int h(y;\mu^{(1)})^{\gamma}(y-\mu^{(1)})g(y|x)dy~g(x)dx.\nonumber
\end{gather} 
In particular, under homogeneous contamination, $B_1$ reduces to 0.
\end{theorem}
\noindent The DP-IPW estimating equation is still biased even if $\nu_1(\varepsilon_1)$ is small. Since we assume that $\nu_1(\varepsilon_1)$ is negligible, $B_1$ is dominant. However, compared to \eqref{eq:IPWbias}, the dominant bias of DP-IPW does not contain $\delta_1$. This implies that the bias of DP-IPW is not strongly affected by the absolute value of outliers. Under homogeneous contamination, the dominant term disappears, so the bias is negligible.

% \begin{align}\label{eq:DPDReq}
%     \sum_{i=1}^n\left[
%         \frac{T_ih(Y_i;\mu)^{\gamma}}{\pi(X_i;\hat{\alpha})}(Y_i-\mu)
%         - \frac{T_i-\pi(X_i;\hat{\alpha})}{\pi(X_i;\hat{\alpha})}
%         \mathbb{E}_{\hat{q}}\left[h(Y_i;\mu)^{\gamma}(Y_i-\mu)|T=1,X\right]
        
%         \left\{m_{1,\mu}(X_i;\hat{\beta}) - \mu m_{0,\mu}(X_i;\hat{\beta})\right\}
%     \right] = 0,
% \end{align} 

\subsection{DR-type Estimator}\label{sec:DP-DR}
Next, we introduce the density power doubly robust (DP-DR) estimator.
The DP-DR estimator is a straight application of the DR M-estimator and defined as a root of the following estimating equation:
\begin{align}\label{eq:DPDReq}
    \sum_{i=1}^n\left[
        \frac{T_ih(Y_i;\mu)^{\gamma}}{\pi(X_i;\hat{\alpha})}(Y_i-\mu)
        - \frac{T_i-\pi(X_i;\hat{\alpha})}{\pi(X_i;\hat{\alpha})}
        \mathbb{E}_{\hat{q}}\left[h(Y;\mu)^{\gamma}(Y-\mu)|T=1,X\right]
    \right] = 0.
\end{align} 
As we have discussed in Section \ref{sec:ATE}, $\mathbb{E}_{\hat{q}}\left[h(Y;\mu)^{\gamma}(Y-\mu)|T=1,X\right]$ is obtained by direct calculation or Monte Carlo approximation based on the parametric OR model $\hat{q}:=q(y|T=1,X;\hat{\beta})$. In the Appendix, we present the explicit forms of $\mathbb{E}_{\hat{q}}\left[h(Y;\mu)^{\gamma}(Y-\mu)|T=1,X\right]$ when $h$ and $q$ are assumed to be Gaussian. The parameter $\beta$ is usually estimated in an outlier-resistant manner: Huber regression \citep[][Chap.7]{huber2004robust}, MM estimator \citep{yohai1987high}, density power regression \citep{basu1998robust,kanamori2015robust}, and $\gamma$-regression \citep{fujisawa2008robust,kawashima2017robust}, for example. 
% Unlike the existing density-power approaches, DP-DR does not multiply the whole estimating equation by the weight $h^{\gamma}$. Instead, we use $h^{\gamma}$ as a multiplicative factor on the first term of \eqref{eq:DPDReq}, which is usual, but incorporate $h^{\gamma}$ inside the conditional expectation on the second term of \eqref{eq:DPDReq}, which is unusual. 

The DP-DR estimator is doubly robust under no contamination as with the general DR M-estimator.
\begin{theorem}\label{thm:dpdr1}
Assume either the true PS or the true OR model is given. Then, if there is no contamination, the DP-DR estimating equation is unbiased.
\end{theorem}

Now, we evaluate the bias of the DP-DR estimating equation under contamination.
\begin{theorem}\label{thm:dpdr2}
Assume that $Y$ is contaminated as \eqref{eq:contam1}.
If the true PS model is given, the expectation of the DP-DR estimating equation is expressed as
\begin{align}\label{eq:dpdrbiasps}
    -\int \varepsilon_1(x)\int h(y;\mu^{(1)})^{\gamma}(y-\mu^{(1)})g(y|x)dy~g(x)dx + \nu_1(\varepsilon_1).
\end{align} 
In particular, under homogeneous contamination, \eqref{eq:dpdrbiasps} reduces to $\nu_1(\varepsilon_1)$.
\noindent If the true OR model is given, the expectation of the DP-DR estimating equation is expressed as
\begin{align}\label{eq:thm4or}
    -\int \varepsilon_1(x)\frac{P(T=1|x)}{\pi(x;\alpha)}\int h(y;\mu^{(1)})^{\gamma}(y-\mu^{(1)})g(y|x)dy~g(x)dx + \nu_1(\varepsilon_1P(T=1|\cdot)/\pi(\cdot;\alpha)).
\end{align}
Under homogeneous contamination, \eqref{eq:thm4or} becomes
\begin{align}\label{eq:dpdrbiasor}
    -\varepsilon_1\int \frac{P(T=1|x)}{\pi(x;\alpha)}\int h(y;\mu^{(1)})^{\gamma}(y-\mu^{(1)})g(y|x)dy~g(x)dx + \nu_1(\varepsilon_1P(T=1|\cdot)/\pi(\cdot;\alpha)).
\end{align}
\end{theorem}

\noindent Assuming that $\pi(\cdot;\alpha)$ is bounded away from 0 and 1, we find that $P(T=1|\cdot)/\pi(\cdot;\alpha)$ is bounded. 
Then, from Assumption \ref{as:Outlier}, $\nu_1(\varepsilon_1P(T=1|\cdot)/\pi(\cdot;\alpha))$ is negligible. As with DP-IPW, the dominant bias is independent of $\delta$, indicating that the influence of outliers is reduced. Unfortunately, DP-DR is still biased in the PS-incorrect and OR-correct case even under homogeneous contamination because the dominant term of \eqref{eq:dpdrbiasor} is not eliminated.

In the OR-correct case, the reason why DP-DR is biased under homogeneous contamination is as follows. Under Assumption 1, the expectation of the DP-DR estimating function becomes
\begin{align}
    &~ \mathbb{E}_{g}\left[\frac{P(T=1|X)}{\pi(X;\alpha)}\left\{\mathbb{E}_{\tilde{g}}[\psi(Y^{(1)};\mu^{(1)})|X] - \mathbb{E}_g[\psi(Y^{(1)};\mu^{(1)})|X]\right\}\right] \nonumber\\
    \approx&~ \mathbb{E}_{g}\left[\frac{P(T=1|X)}{\pi(X;\alpha)}\left\{(1-\varepsilon_1)\mathbb{E}_{g}[\psi(Y^{(1)};\mu^{(1)})|X] - \mathbb{E}_g[\psi(Y^{(1)};\mu^{(1)})|X]\right\}\right], \nonumber
\end{align}
where we denote the density power estimating function by $\psi$. In the last formula, it is found that the terms in the curly brackets do not cancel because the first term is reduced by $1-\varepsilon_1$. Based on this consideration, we propose a bias-corrected version of DP-DR, called the $\varepsilon$DP-DR estimator. $\varepsilon$DP-DR is designed to cancel the dominant bias under homogeneous contamination. 
The $\varepsilon$DP-DR estimator is a root of the following estimating equation:
\begin{align}\label{eq:eDPDReq}
    \sum_{i=1}^n\left[
        \frac{T_ih(Y_i;\mu)^{\gamma}}{\pi(X_i;\hat{\alpha})}(Y_i-\mu)
        - \frac{T_i-\pi(X_i;\hat{\alpha})}{\pi(X_i;\hat{\alpha})}(1-\hat\varepsilon_1)\mathbb{E}_{\hat{q}}\left[h(Y;\mu)^{\gamma}(Y-\mu)|T=1,X\right]
    \right] = 0,
\end{align}
where $\hat\varepsilon_1$ is a consistent estimator of the expected contamination ratio $\overline{\varepsilon}_1=\int\varepsilon_1(x)g(x)dx$. 
$\hat\varepsilon_1$ can be obtained simultaneously with the parametric OR model by the unnormalized modeling with the density power score \citep{kanamori2015robust}, for example. While DP-DR is a special case of the DR M-estimator, $\varepsilon$DP-DR goes beyond this framework by the bias correction.
Under no contamination, the $\varepsilon$DP-DR estimating equation is asymptotically identical to the DP-DR estimating equation. The $\varepsilon$DP-DR estimating equation is also biased under heterogeneous contamination; however, the bias takes a different form.
\begin{corollary} \label{cor:dpdr}
If the true PS model is given, the expectation of the $\varepsilon$DP-DR estimating equation is equal to \eqref{eq:dpdrbiasps}.
If the true OR model is given, the expectation of the $\varepsilon$DP-DR estimating equation is expressed as
\begin{align}\label{eq:biaseDPDR}
    \mathbb{E}_g\left[
        (\overline{\varepsilon}_1-\varepsilon_1(X))\frac{P(T=1|X)}{\pi(X;\alpha)}\mathbb{E}_g[h(Y^{(1)};\mu^{(1)})^{\gamma}(Y^{(1)}-\mu^{(1)})|X]
    \right] + \nu_1(\varepsilon_1P(T=1|\cdot)/\pi(\cdot;\alpha)).
\end{align}
The first term disappears under homogeneous contamination.
\end{corollary}

\begin{proof}
The derivation is the same as that of Theorem \ref{thm:dpdr2}.
If $\varepsilon_1(X)$ is constant $\varepsilon_1$, the first term disappears because $\overline{\varepsilon}_1=\varepsilon_1\int g(x)dx=\varepsilon_1$.
\end{proof}

Similar to \eqref{eq:dpdrbiasor}, the second term of \eqref{eq:biaseDPDR} is approximately zero if we assume that $\pi(\cdot;\alpha)$ is bounded away from 0 and 1. 

\paragraph{Remark}
One may find that ``$\varepsilon(X)$''DP-DR would work better than $\varepsilon$DP-DR under heterogeneous contamination. In fact, the bias \eqref{eq:biaseDPDR} will disappear if we replace $\overline\varepsilon$ with $\varepsilon(X)$. 
However, it is necessary to model $\varepsilon(X)$ correctly for consistent estimation of ``$\varepsilon(X)$''DP-DR. To the best of our knowledge, no easy method is available for this purpose.

\subsection{Summary}\label{sec:summary}
We have proposed three types of outlier resistant semiparametric estimators: DP-IPW, DP-DR, and $\varepsilon$DP-DR. Table \ref{tab:summary} shows the bias of the estimating equations under the conditions discussed above.
Under heterogeneous contamination, all estimators are biased, but the biases are hardly influenced by the absolute value of outliers. Furthermore, as discussed in Section \ref{sec:if}, outliers have negligible influence if the contamination ratio is sufficiently small. $\varepsilon$DP-DR improves DP-DR in the OR-correct case under homogeneous contamination, but we continue to discuss DP-DR for three reasons: the contamination ratio is sometimes hard to estimate, the bias \eqref{eq:dpdrbiasor} is not serious if $\pi(X;\alpha)$ is close to $P(T=1|X)$, and the simulation results presented in Section \ref{sec:monte} indicate that DP-DR remains better than the existing methods even in the OR-correct case. 

\begin{table}[ht]
    \centering
    \begin{tabular}{c|c|ccc}
        \hline
        Contamination & model & DP-IPW & DP-DR & $\varepsilon$DP-DR \\
        \hline\hline
        No contamination 
            & PS-correct & $0$ & $0$ & $0$ \\
            & OR-correct & - & $0$ & $0$ \\ \hline
        homogeneous: $\varepsilon$ 
            & PS-correct & $\approx0$ & $\approx0$ & $\approx0$ \\
            & OR-correct & - & $\approx\varepsilon\mathbb{E}[\phi(X)]$ & $\approx0$ \\ \hline
        heterogeneous: $\varepsilon(X)$ 
            & PS-correct & $\approx\mathbb{E}[\varepsilon(X)\phi(X)]$ & $\approx\mathbb{E}[\varepsilon(X)\phi(X)]$ & $\approx\mathbb{E}[\varepsilon(X)\phi(X)]$ \\
            & OR-correct & - & $\approx\mathbb{E}[\varepsilon(X)\phi(X)]$ & $\approx\mathbb{E}[(\overline{\varepsilon}-\varepsilon(X))\phi(X)]$ \\ \hline
    \end{tabular}
    \caption{Summary of the biases of the proposed estimating equations. The function $\phi(X)$ differs cell-by-cell. PS-correct means that the PS model is correctly specified and the OR model may not be; OR-correct means the opposite.}
    \label{tab:summary}
\end{table}
 
%%%%%%%%%%%%%%%%%%%%%%%%%%%%%%%%%%%%%%%%%%%%%%%%%%%%%%%%%%%%%%%%%%%%%%%%%%%%%%%%%%%%%%%%%%%%%%%%%%%%%%%%%%%%%%%%%%%%%%%%%%%%

\section{Influence-function-based Analysis of Outlier Resistance}\label{sec:if}
As discussed in the previous section, the proposed estimators suffer less from outliers compared with ordinary estimators from the viewpoint of the unbiasedness of the estimating equation. In this section, we demonstrate that they are outlier-resistant from the viewpoint of the IF.

Here, we briefly review the IF for the univariate M-estimator and expand it to evaluate our estimators.
Let $G$ be the distribution of $Z\in\mathbb{R}$, and let $T(G)$ be a functional of G, which is the parameter of interest. The IF of $T(G)$ is defined as
\begin{align}
    IF(z_0;G):=\lim_{\varepsilon\rightarrow 0}\frac{T((1-\varepsilon)G+\varepsilon\Delta_{z_0})-T(G)}{\varepsilon} = \left.\frac{\partial}{\partial\varepsilon}\{T((1-\varepsilon)G+\varepsilon\Delta_{z_0})-T(G)\}\right|_{\varepsilon=0},
\end{align}
where $\Delta_{z_0}$ is a degenerate distribution at $z_0$.
We also see that the latent bias $T((1-\varepsilon)G+\varepsilon \Delta_{z_0})-T(G)$ can be approximated by $\varepsilon IF(z_0;G)$. Therefore, the behavior of the IF approximates that of the latent bias.
In the population, the M-estimator $T_M(G)$ satisfies $\int\psi(z,T_M(G))dG(z)=0$. Then, the IF for $T_M(G)$ is obtained by differentiating $\int\psi(z,T_{M}((1-\varepsilon)G+\varepsilon\Delta_{z_0})d\{(1-\varepsilon)G+\varepsilon\Delta_{z_0}\}(z)=0$ with respect to $\varepsilon$. This yields 
\begin{align}
    IF(z_0;G) = -\mathbb{E}\left[\left.\frac{\partial}{\partial\eta}\psi(Z,\eta)\right|_{\eta=T_M(G)}\right]^{-1}\psi(z_0,T_M(G)).
\end{align}
The function $\psi$ is said to have a redescending property if $\psi(z_0,T_M(G))$ approaches zero as the outlier $|z_0|$ increases. Therefore, when $\psi$ has the redescending property and $z_0$ is an outlier, the latent bias is sufficiently small. This is favorable for outlier resistance.

Since $\varepsilon_1$ is dependent on $X$, we cannot apply the IF directly to our estimators. To overcome this issue, we consider the IF with fixed covariates $\{X_i\}_{i=1}^n$; this approach is similar to the fixed carrier model in \citet[][Chap.6]{hampel2011robust}.
Consider the following estimating equation:
\begin{align} \label{eq:condMEE}
    \frac{1}{n}\sum_{i=1}^n\left.\mathbb{E}_{\tilde{g}}\left[\psi(Y,T,X_i;\mu)\right|X_i\right] = 0.
\end{align}
If the fixed sample $\{X_i\}_{i=1}^n$ consists of i.i.d. observations, then the left-hand side of \eqref{eq:condMEE} converges to $\mathbb{E}_{\tilde{g}}[\psi(Y,T,X;\mu)]$ as $n\rightarrow\infty$. Let $\tilde{\mu}^{(1)}_n$ denote a root of \eqref{eq:condMEE}, and let $\tilde{\mu}^{(1)}$ be a root of $\mathbb{E}_{\tilde{g}}[\psi(Y,T,X;\mu)]$. Then, $\tilde{\mu}^{(1)}_n$ also converges to $\tilde{\mu}^{(1)}$. Therefore, $\tilde{\mu}^{(1)}_n$ shows roughly the same behavior as that of the target estimator $\tilde{\mu}^{(1)}$. The contaminated density $\tilde{g}$ is defined as \eqref{eq:contam1}, and $\delta_1(y|X_i)$ is assumed to be Dirac's delta at $y_0$. The IF of $T_n(\tilde{G})$ at $X_i$ is obtained by differentiating \eqref{eq:condMEE} with respect to $\varepsilon_1(X_i)$ at $\varepsilon_1(X_i)=0$. 

Because of the space, we discuss only the $\varepsilon$DP-DR.
Assume that $\overline{\varepsilon}_1 = \frac{1}{n}\sum_{i=1}^n\varepsilon_1(X_i)$, then the IF of $\varepsilon$DP-DR is
\begin{align} \label{eq:ifedpdr}
    &-\mathbb{E}_{g}\left[\left.\left.\frac{\partial\psi}{\partial\mu}\right|_{\mu=\mu^{(1)}_n}\right|X_i\right]^{-1}\left[\frac{P(T=1|X_i)}{\pi(X_i;\alpha)}h(y_{0}-\mu_n^{(1)})^{\gamma}(y_{0}-\mu^{(1)}_n)\right.\nonumber\\
    &\left.-\frac{n-1}{n}\frac{P(T=1|X_i)-\pi(X_i;\alpha)}{\pi(X_i;\alpha)}\mathbb{E}_{\hat{q}}\left[h(Y;\mu_n^{(1)})^{\gamma}(Y-\mu_n^{(1)})|T=1,X\right]\right].
\end{align}
In the PS-correct case, the second term in square brackets is equal to zero, and then the IF tends to zero as $|y_0|\rightarrow\infty$. In the OR-correct case, the second term does not disappear. Considering the limit of $|y_0|\rightarrow\infty$, the IF converges to
\begin{align} \label{eq:ifdpdr2}
    \frac{n-1}{n}\mathbb{E}_{g}\left[\left.\left.\frac{\partial\psi}{\partial\mu}\right|_{\mu=\mu^{(1)}_n}\right|X_i\right]^{-1}&\left[ \frac{P(T=1|X_i)-\pi(X_i;\alpha)}{\pi(X_i;\alpha)}\mathbb{E}_{\hat{q}}[h(Y;\mu_n^{(1)})^{\gamma}(Y-\mu^{(1)}_n)|T=1,X_i]\right].
\end{align}
Thus, the $\varepsilon$DP-DR estimator has the redescending property only in the PS-correct case. In the OR-correct case, the influence cannot be eliminated, but the IF tends to a constant when $|y_0|$ tends to infinity, implying that the influence of the outlier is not serious. DP-DR has an IF similar to that of $\varepsilon$DP-DR, and DP-IPW has an IF similar to that of $\varepsilon$DP-DR whose PS is correct. The derivations of all IFs are presented in the Appendix.

Under homogeneous contamination, the ordinary IF is applicable, and
we can see that the proposed estimators have the redescending property in the PS-correct case. Besides, $\varepsilon$DP-DR has the redescending property even in the OR-correct case; this result is consistent with Corollary \ref{cor:dpdr}.
The IF-based analysis under homogeneous contamination is presented in the Appendix.

%%%%%%%%%%%%%%%%%%%%%%%%%%%%%%%%%%%%%%%%%%%%%%%%%%%%%%%%%%%%%%%%%%%%%%%%%%%%%%%%%%%%%%%%%%%%%%%%%%%%%%%%%%%%%%%%%%%%%%%%%%%%

\section{Asymptotic Properties}\label{sec:asymp}
We discuss the asymptotic properties of the $\varepsilon$DP-DR estimator. For the other proposed estimators, we obtain similar results with small changes. The asymptotic properties can be obtained in a manner similar to that described in \cite{hoshino2007doubly}. 
Assume that the PS and OR models are regular and are estimated consistently if the models are correctly specified.
Furthermore, the contamination ratio $\varepsilon_1$ is known. Note that when the contamination ratio is consistently estimated simultaneously with the OR model by \cite{kanamori2015robust}, we can replace $\beta$ with $(\varepsilon_1,\beta^T)^T$ in the following discussion.

We write \eqref{eq:eDPDReq} as $\frac{1}{n}\sum_{i=1}^n\psi_i(\mu;\hat\alpha,\hat\beta)$, and let $\frac{1}{n}\sum_{i=1}^n s^{PS}_i(\alpha)=0$ and $\frac{1}{n}\sum_{i=1}^n s^{OR}_i(\beta)=0$ be the estimating equations for the PS and OR models, respectively. Let $\lambda = (\mu,\alpha^T,\beta^T)^T$ be the parameter vector, and let the full estimating equation be defined as
\begin{align}\label{eq:fullee}
    \sum_{i=1}^nS_i(\lambda)=\sum_{i=1}^n\left(\begin{array}{c} \psi_i(\mu;\alpha,\beta) \\ s^{PS}_i(\alpha) \\ s^{OR}_i(\beta) \end{array}\right) = \mathbf{0}.
\end{align}
Let $\lambda^* = (\mu^*,\alpha^{*T},\beta^{*T})^T$ be a root of \eqref{eq:fullee} in population. Note that, in this section, $*$ does not necessarily mean that the model is correctly specified. With the results presented in \citet[][Chap.5]{van2000asymptotic}, the following theorem holds under some regularity conditions.

\begin{theorem} \label{thm:asymp}
Under the regularity conditions presented in the Appendix, the following asymptotic properties hold:
% \begin{spacing}{1}
\begin{align}
    \hat\lambda &\overset{p}{\rightarrow} \lambda^*,\\
    \sqrt{n}(\hat\lambda-\lambda^*)
        &\overset{d}{\rightarrow}\mathcal{N}\left(
            \mathbf{0},\mathbf{V}^{\tilde{g}}(\lambda^*) 
        \right),
    % \mathbf{V}^{\tilde{g}}(\lambda^*) 
    %     &=~ \mathbf{J}^{\tilde{g}}(\lambda^*)^{-1}\mathbf{K}^{\tilde{g}}(\lambda^*)\{\mathbf{J}^{\tilde{g}}(\lambda^*)^T\}^{-1}, \\
    % \mathbf{J}^{\tilde{g}}(\lambda^*) 
    %     &=~ \mathbb{E}_{\tilde{g}}\left[\partial S_i(\lambda^*)/\partial\lambda^T\right], \\
    % \mathbf{K}^{\tilde{g}}(\lambda^*) 
    %     &=~ \mathbb{E}_{\tilde{g}}\left[S_i(\lambda^*)S_i(\lambda^*)^T\right],
\end{align}
where $\mathbf{V}^{\tilde{g}}(\lambda^*) = \mathbf{J}^{\tilde{g}}(\lambda^*)^{-1}\mathbf{K}^{\tilde{g}}(\lambda^*)\{\mathbf{J}^{\tilde{g}}(\lambda^*)^T\}^{-1}$, $\mathbf{J}^{\tilde{g}}(\lambda^*) = \mathbb{E}_{\tilde{g}}\left[\partial S_i(\lambda^*)/\partial\lambda^T\right]$, and $\mathbf{K}^{\tilde{g}}(\lambda^*) = \mathbb{E}_{\tilde{g}}\left[S_i(\lambda^*)S_i(\lambda^*)^T\right]$.
% \end{spacing}
\end{theorem}

By using this and applying the results presented in Section \ref{sec:DP-DR}, we find that the limit $\mu^*$ is in the neighborhood of $\mu^{(1)}$.
\begin{theorem}\label{thm:consist}
Let $\lambda^{**}=(\mu^{(1)},\alpha^{*T},\beta^{*T})^T$ and assume that $\mathbf{J}^{\tilde{g}}_{11}(\lambda)$ is nonzero within the interval $[\lambda^*, \lambda^{**}]$.
Under Assumption \ref{as:Outlier} and homogeneous contamination, if either the PS or the OR model is correct, it then holds that
\begin{align}
    \mu^* = \mu^{(1)} + \mathcal{O}(\nu_1(\phi)),
\end{align}
where $\phi(\cdot)=\varepsilon_1$ (constant) in the PS-correct case and $\phi(\cdot)=\varepsilon_1P(T=1|\cdot)/\pi(\cdot;\alpha)$ in the OR-correct case.
\end{theorem}
\noindent The proof of Theorem \ref{thm:consist} and further discussions on the asymptotic variance are available in the Appendix.

%%%%%%%%%%%%%%%%%%%%%%%%%%%%%%%%%%%%%%%%%%%%%%%%%%%%%%%%%%%%%%%%%%%%%%%%%%%%%%%%%%%%%%%%%%%%%%%%%%%%%%%%%%%%%%%%%%%%%%%%%%%%

\section{Monte-Carlo Simulation}\label{sec:monte}
We conduct Monte-Carlo simulations to evaluate the performance of the proposed estimators. Our methods are compared with the naive IPW and DR estimators and some existing outlier-resistant methods (\citeauthor{firpo2007efficient}, \citeyear{firpo2007efficient}; \citeauthor{zhang2012causal}, \citeyear{zhang2012causal}; \citeauthor{diaz2017efficient}, \citeyear{diaz2017efficient}; \citeauthor*{sued2020robust}, \citeyear{sued2020robust}).
Since these methods focus on the median of the potential outcome, they are resistant to outliers at a certain level; but the median-based methods are not so resistant to heavy contamination. To the best of our knowledge, no method other than the proposed method has more outlier resistance than the median. Firpo's IPW estimator \citep{firpo2007efficient} is defined as 
\begin{align}
    \hat{\mu}_{\mathrm{Firpo}} = \mathrm{arg}\min_{\mu}\sum_{i=1}^n\frac{T_i}{\pi(X_i;\hat{\alpha})}(Y_i-\mu)(0.5 - \mathbb{I} (Y_i \le \mu)),
\end{align}
where the function $\mathbb{I}$ is an indicator function. 
Zhang's IPW median \citep{zhang2012causal} is based on the IPW-empirical distribution. 
% \begin{align}
%     \hat{F}_{\mathrm{IPW}}(y) = \left.\left(\sum_{i=1}^{n}\frac{T_i \mathbb{I} (Y_i \le y)}{\pi(X_i;\hat{\alpha})}\right)\right/ \left(\sum_{i=1}^{n}\frac{T_i}{\pi(X_i;\hat{\alpha})}\right),
% \end{align}
% and the median is estimated as $y_0$ such that $\hat{F}_{IPW}(y_0)=0.5$. 
Firpo's IPW and Zhang's IPW are almost equivalent except for a slight difference in their computation.
Zhang's and Sued's DR methods (\citeauthor{zhang2012causal}, \citeyear{zhang2012causal}; \citeauthor*{sued2020robust}, \citeyear{sued2020robust}) estimate the empirical distribution in a doubly robust way. They incorporate an IPW-type estimator into the first term. The remaining term of Zhang's DR is based on the Gaussian cumulative distribution function of $Y$ given $X$. In contrast, Sued's DR constructs the remaining term in a nonparametric manner. Diaz's DR median \citep{diaz2017efficient} is a different approach; it employs the targeted maximum likelihood estimator (TMLE) \citep{van2006targeted}. We implemented our methods, Zhang's IPW/DR, and Sued's DR in R. For Firpo's IPW and TMLE, we used the {\it causalquantile} package (https://github.com/idiazst/causalquantile; Updated on 31 Aug 2017).

\subsection{Numerical Algorithm for the Proposed Methods}
Since the proposed estimating equations cannot be solved explicitly, we develop an iterative algorithm. Various algorithms are available, but we propose a standard algorithm for M-estimators \citep{huber2004robust,hampel2011robust}. Detailed algorithm is available in the Appendix. Hereafter, we suppose $h$ and $q$ are Gaussian. We also provide explicit updating formulae in this case. Note that some additional parameters of $h$ should be estimated in a roughly unbiased and outlier-resistant way.

\subsection{Simulation Model}
We simulated random observations based on a simple causal setting. The confounders $(X_1,X_2)$ were independently drawn from a Gaussian or uniform distribution with mean zero and unit variance. The treatment $T$ was assigned along with the conditional probability $P(T=1|X_1,X_2)$ that was defined as a sigmoid function of $0.8 X_1+0.2 X_2$. The potential outcomes $(Y^{(1)},Y^{(0)})$ were generated according to a linear function of $(X_1,X_2)$ with Gaussian error: $Y^{(1)} = \mu^{(1)} + 1.2 X_1 + 0.3 X_2 + e$ and $Y^{(0)} = \mu^{(0)} + 1.2 X_1 + 0.3 X_2 + e$. $\mu^{(1)}$ and $\mu^{(0)}$ were set to 3 and 0, respectively. The standard deviation (SD) of $e$ was set to $\sqrt{0.72}$; then, $\mathrm{SD}[Y^{(1)}]=\mathrm{SD}[Y^{(0)}]=1.5$. When the confounders were not Gaussian, the potential outcomes were not Gaussian. The observed outcome $Y$ was defined as $Y=TY^{(1)}+(1-T)Y^{(1)}$ under no contamination. Outliers were drawn from $\mathcal{N}(\mu^{(t)}+10\sigma^{(t)},1)$, with $\sigma^{(t)}=\mathrm{SD}[Y^{(t)}]=1.5$. For the homogeneous contamination settings, the contamination ratio was set to be a constant $\varepsilon_t$. For the heterogeneous contamination settings, the contamination ratio was set to be $1.5\varepsilon_t$ if $X_1+X_2 \le 0$ and $0.5\varepsilon_t$ if $X_1+X_2 > 0$. The average contamination ratio is set to $\varepsilon_t \in \{0,0.05,0.1,0.2\}$. Then, the observations of $Y$ were randomly replaced with outliers according to the contamination ratio. The sample size was fixed to $n=100$ throughout the Monte Carlo simulations. 
Furthermore, we generated datasets in which the outcome follows a symmetric and heavy-tailed distribution. We drew the error term of $Y^{(t)}$ from the standard Cauchy distribution instead of inserting outliers.

\subsection{Results}
First, we performed a comparative study. The potential mean $\mu^{(1)}$ was estimated using the proposed and comparative methods. In this experiment, we used all settings illustrated in the previous section.
The propensity score was estimated by logistic regression. The parametric OR was conducted in two ways: Gaussian MLE with non-outliers or unnormalized Gaussian modeling (the tuning parameter was set to 0.5) \citep{kanamori2015robust}. For the DR estimators, we investigated three patterns of model misspecification: PS-correct/OR-correct, PS-correct/OR-incorrect, and PS-incorrect/OR-correct. For the model-correct case, we included an intercept and $(X_1,X_2)$ as covariates. For the model-incorrect case, we included only an intercept and $X_2$.
We performed 10,000 simulations for every setting and method. Tables \ref{tab:exp1ipw} and \ref{tab:exp1dr} show the results of the comparative study when the covariates were Gaussian and the OR for the DR-type estimators was the Gaussian MLE with non-outliers. The estimation error was measured by the root mean square error (RMSE). The mean and SD of all estimates, the mean computation time, and the results for the other settings are provided in the Appendix.
In Table \ref{tab:exp1ipw}, the naive IPW estimator had a significantly larger RMSE under contamination. Both the median-based methods and DP-IPW dramatically reduced the RMSE. As the contamination ratio increased, the RMSE increased. The RMSE tended to be larger for heterogeneous contamination than for homogeneous contamination. 
When the optimal $\gamma$ was chosen, the proposed method outperformed the comparative methods and had the smallest RMSE for all settings.
Looking at Table \ref{tab:exp1dr}, the results for the DR-type estimators were similar to those for the IPW estimators. The proposed method with a proper $\gamma$ outperformed the comparative methods and had the smallest RMSE in all settings. DP-DR and $\varepsilon$DP-DR performed similarly, although $\varepsilon$DP-DR was slightly superior in many settings. 
Among the median-based methods, TMLE performed better, but it took much more time than the other methods, including the proposed methods, and occasionally ($<1\%$) failed to converge.

\begin{spacing}{1}
% latex table generated in R 4.0.3 by xtable 1.8-4 package
% Thu Jun 03 09:46:16 2021
\begin{landscape}
\begin{table}[ht]
\centering
\footnotesize
\begin{tabular}{llccccccc}
  \hline
  \multicolumn{2}{l}{} & {\bf No contam.} & \multicolumn{3}{c}{\bf Homogeneous} & \multicolumn{3}{c}{\bf Heterogeneous} \\ 
  &  &  & $\varepsilon=0.05$ & $\varepsilon=0.10$ & $\varepsilon=0.20$ & $\varepsilon=0.05$ & $\varepsilon=0.10$ & $\varepsilon=0.20$ \\ 
    \hline
IPW & Naive & 0.222 & 0.957 & 1.683 & 3.153 & 0.993 & 1.752 & 3.253 \\ 
   & median (Firpo) & 0.257 & 0.294 & 0.367 & 0.649 & 0.306 & 0.409 & 0.769 \\ 
   & median (Zhang-IPW) & 0.257 & 0.294 & 0.367 & 0.649 & 0.306 & 0.409 & 0.769 \vspace{2mm}\\ 
   & DP-IPW ($\gamma=0.1$) & 0.218 & 0.276 & 0.531 & 2.263 & 0.293 & 0.609 & 2.377 \\ 
   & DP-IPW ($\gamma=0.5$) & 0.227 & 0.249 & 0.272 & 0.639 & 0.245 & 0.287 & 0.726 \\ 
   & DP-IPW ($\gamma=1.0$) & 0.261 & 0.271 & 0.275 & 0.413 & 0.262 & 0.281 & 0.498 \\ 
  \hline
\end{tabular}
\caption{Results of the comparative study of the IPW-type estimators. Each figure is RMSE between each estimate and the true value. The covariates $X$ were generated from Gaussian distributions.}
\label{tab:exp1ipw}
\end{table}
\end{landscape}

\begin{landscape}
\begin{table}
\centering
\footnotesize
\begin{tabular}{llccccccc}
  \hline
 &  & {\bf No contam.} & \multicolumn{3}{c}{\bf Homogeneous} & \multicolumn{3}{c}{\bf Heterogeneous} \\ 
 &  &  & $\varepsilon=0.05$ & $\varepsilon=0.10$ & $\varepsilon=0.20$ & $\varepsilon=0.05$ & $\varepsilon=0.10$ & $\varepsilon=0.20$ \\ 
  \hline      
  DR(T/T) & Naive & 0.184 & 0.957 & 1.684 & 3.154 & 0.997 & 1.758 & 3.265 \\ 
   & median (Zhang-DR) & 0.239 & 0.317 & 0.391 & 0.733 & 0.330 & 0.452 & 0.905 \\ 
   & median (Sued) & 0.238 & 0.316 & 0.388 & 0.693 & 0.329 & 0.450 & 0.869 \\ 
   & median (TMLE) & 0.237 & 0.280 & 0.359 & 0.603 & 0.295 & 0.402 & 0.701\vspace{2mm} \\ 
   & DP-DR ($\gamma=0.1$) & 0.183 & 0.302 & 0.564 & 2.262 & 0.318 & 0.649 & 2.394 \\ 
   & DP-DR ($\gamma=0.5$) & 0.202 & 0.285 & 0.326 & 0.697 & 0.274 & 0.349 & 0.834 \\ 
   & DP-DR ($\gamma=1.0$) & 0.240 & 0.288 & 0.307 & 0.524 & 0.287 & 0.336 & 0.669\vspace{2mm} \\ 
   & $\varepsilon$DP-DR ($\gamma=0.1$) & 0.183 & 0.296 & 0.554 & 2.255 & 0.314 & 0.636 & 2.385 \\ 
   & $\varepsilon$DP-DR ($\gamma=0.5$) & 0.202 & 0.264 & 0.302 & 0.669 & 0.271 & 0.323 & 0.793 \\ 
   & $\varepsilon$DP-DR ($\gamma=1.0$) & 0.240 & 0.287 & 0.299 & 0.513 & 0.286 & 0.335 & 0.648\vspace{2mm} \\ 
  DR(T/F) & Naive & 0.237 & 0.963 & 1.686 & 3.156 & 1.001 & 1.758 & 3.262 \\ 
   & median (Zhang-DR) & 0.275 & 0.342 & 0.408 & 0.741 & 0.350 & 0.465 & 0.912 \\ 
   & median (Sued) & 0.275 & 0.342 & 0.407 & 0.699 & 0.350 & 0.464 & 0.872 \\ 
   & median (TMLE) & 0.242 & 0.284 & 0.363 & 0.622 & 0.297 & 0.404 & 0.719\vspace{2mm} \\ 
   & DP-DR ($\gamma=0.1$) & 0.237 & 0.314 & 0.561 & 2.267 & 0.330 & 0.644 & 2.393 \\ 
   & DP-DR ($\gamma=0.5$) & 0.247 & 0.319 & 0.349 & 0.714 & 0.319 & 0.361 & 0.839 \\ 
   & DP-DR ($\gamma=1.0$) & 0.280 & 0.334 & 0.347 & 0.581 & 0.329 & 0.372 & 0.709\vspace{2mm} \\ 
   & $\varepsilon$DP-DR ($\gamma=0.1$) & 0.237 & 0.311 & 0.557 & 2.264 & 0.328 & 0.640 & 2.388 \\ 
   & $\varepsilon$DP-DR ($\gamma=0.5$) & 0.247 & 0.317 & 0.344 & 0.694 & 0.313 & 0.356 & 0.817 \\ 
   & $\varepsilon$DP-DR ($\gamma=1.0$) & 0.280 & 0.333 & 0.338 & 0.551 & 0.327 & 0.369 & 0.708\vspace{2mm} \\ 
  DR(F/T) & Naive & 0.181 & 0.879 & 1.591 & 3.026 & 0.826 & 1.490 & 2.813 \\ 
   & median (Zhang-DR) & 0.237 & 0.263 & 0.316 & 0.503 & 0.269 & 0.337 & 0.548 \\ 
   & median (Sued) & 0.236 & 0.272 & 0.346 & 0.599 & 0.277 & 0.364 & 0.627 \\ 
   & median (TMLE) & 0.234 & 0.260 & 0.309 & 0.478 & 0.265 & 0.328 & 0.522\vspace{2mm} \\ 
   & DP-DR ($\gamma=0.1$) & 0.182 & 0.192 & 0.345 & 2.057 & 0.191 & 0.299 & 1.681 \\ 
   & DP-DR ($\gamma=0.5$) & 0.199 & 0.206 & 0.218 & 0.366 & 0.203 & 0.209 & 0.283 \\ 
   & DP-DR ($\gamma=1.0$) & 0.230 & 0.232 & 0.239 & 0.273 & 0.230 & 0.233 & 0.242\vspace{2mm} \\ 
   & $\varepsilon$DP-DR ($\gamma=0.1$) & 0.182 & 0.193 & 0.381 & 2.207 & 0.194 & 0.335 & 1.839 \\ 
   & $\varepsilon$DP-DR ($\gamma=0.5$) & 0.199 & 0.203 & 0.208 & 0.376 & 0.203 & 0.212 & 0.318 \\ 
   & $\varepsilon$DP-DR ($\gamma=1.0$) & 0.230 & 0.230 & 0.231 & 0.243 & 0.231 & 0.237 & 0.260 \\ 
   \hline
\end{tabular}
\caption{Results of the comparative study of the IPW-type estimators. Each figure is RMSE between each estimate and the true value. The covariates $X$ were generated from Gaussian distributions, and the outcome regression was obtained by the Gaussian MLE using non-outliers. The characters "T" and "F" denote the correct and the incorrect modeling, respectively.
}
\label{tab:exp1dr}
\end{table}
\end{landscape}

\end{spacing}

Table \ref{tab:exp2cauchy} shows the RMSE of each method on the data with Cauchy error. As well as the above experiments, the proposed method performed better than the comparative methods. In this setting, we only used the unnormalized Gaussian modeling for OR for the DR-type estimators. Only in the PS-correct/OR-incorrect case, the median (TMLE) performed slightly better than the proposed method.

% \begin{spacing}{1}
% latex table generated in R 4.0.3 by xtable 1.8-4 package
% Wed Jun 23 15:34:39 2021
\begin{table}[ht]
\centering
\begin{tabular}{llcc}
  \hline
 &  & \multicolumn{2}{l}{\bf Distribution of $X$} \\
 &  & Gaussian & Uniform \\ 
  \hline
IPW(T/-) & Naive & 274.024 & 246.118 \\ 
   & median (Firpo) & 0.414 & 0.438 \\ 
   & median (Zhang-IPW) & 0.414 & 0.438\vspace{2mm}\\ 
   & DP-IPW ($\gamma=0.1$) & 0.443 & 0.425 \\ 
   & DP-IPW ($\gamma=0.5$) & 0.367 & 0.363 \\ 
   & DP-IPW ($\gamma=1.0$) & 0.380 & 0.383\vspace{2mm}\\ 
  DR(T/T) & Naive & 275.447 & 247.011 \\ 
   & median (Zhang-DR) & 0.415 & 0.431 \\ 
   & median (Sued) & 0.408 & 0.430 \\ 
   & median (TMLE) & 0.392 & 0.425\vspace{2mm}\\ 
   & DP-DR ($\gamma=0.1$) & 0.501 & 0.420 \\ 
   & DP-DR ($\gamma=0.5$) & 0.363 & 0.356 \\ 
   & DP-DR ($\gamma=1.0$) & 0.372 & 0.374\vspace{2mm}\\ 
   & $\varepsilon$DP-DR ($\gamma=0.1$) & 0.487 & 0.420 \\ 
   & $\varepsilon$DP-DR ($\gamma=0.5$) & 0.361 & 0.355 \\ 
   & $\varepsilon$DP-DR ($\gamma=1.0$) & 0.370 & 0.374\vspace{2mm}\\ 
  DR(T/F) & Naive & 275.446 & 247.011 \\ 
   & median (Zhang-DR) & 0.456 & 0.443 \\ 
   & median (Sued) & 0.436 & 0.441 \\ 
   & median (TMLE) & 0.394 & 0.427\vspace{2mm}\\ 
   & DP-DR ($\gamma=0.1$) & 0.514 & 0.431 \\ 
   & DP-DR ($\gamma=0.5$) & 0.404 & 0.369 \\ 
   & DP-DR ($\gamma=1.0$) & 0.418 & 0.389\vspace{2mm}\\ 
   & $\varepsilon$DP-DR ($\gamma=0.1$) & 0.503 & 0.430 \\ 
   & $\varepsilon$DP-DR ($\gamma=0.5$) & 0.399 & 0.368 \\ 
   & $\varepsilon$DP-DR ($\gamma=1.0$) & 0.412 & 0.388\vspace{2mm}\\ 
  DR(F/T) & Naive & 263.629 & 177.037 \\ 
   & median (Zhang-DR) & 0.390 & 0.429 \\ 
   & median (Sued) & 0.373 & 0.400 \\ 
   & median (TMLE) & 0.389 & 0.429\vspace{2mm}\\ 
   & DP-DR ($\gamma=0.1$) & 0.390 & 0.401 \\ 
   & DP-DR ($\gamma=0.5$) & 0.358 & 0.376 \\ 
   & DP-DR ($\gamma=1.0$) & 0.364 & 0.393\vspace{2mm}\\ 
   & $\varepsilon$DP-DR ($\gamma=0.1$) & 0.377 & 0.385 \\ 
   & $\varepsilon$DP-DR ($\gamma=0.5$) & 0.328 & 0.338 \\ 
   & $\varepsilon$DP-DR ($\gamma=1.0$) & 0.334 & 0.351 \\ 
   \hline
\end{tabular}
\caption{Results of the comparative study using the heavy-tailed data. Each figure is RMSE between each estimate and the true value. The covariates $X$ were generated from Gaussian or uniform distributions. The OR model for the DR-type estimators were obtained by the unnormalized Gaussian modeling. The characters "T" and "F" denote the correct and the incorrect modeling, respectively.}
\label{tab:exp2cauchy}
\end{table}

% \end{spacing}

Next, we conduct a $\gamma$-sensitivity study. $\mu^{(1)}$ was estimated by the proposed methods with different $\gamma$s. $X$ had a Gaussian distribution, and the contamination ratio varied in $\{0,0.05,0.1,0.2\}$ under homogeneous contamination. For the DR-type estimators, the outcome regression was performed by the Gaussian MLE using nonoutliers.
We simulated 10,000 datasets for every setting and method. Table \ref{tab:exp2} shows the results of the $\gamma$-sensitivity study. As in the comparative study, when the ratio of outliers increased, the bias increased. Larger $\gamma$ resulted in increased variance. When the contamination ratio was small, it was sufficient to use a small $\gamma$ such as $\gamma=0.1$ or $0.2$ to remove the adverse effect of outliers. Even in highly contaminated cases, $\gamma > 1.0$ was not needed. Comparing the estimates of DP-DR and $\varepsilon$DP-DR in the PS-incorrect/OR-correct case, it can be found that the DP-DR estimates were biased especially when $\varepsilon$ was large, and contrarily, the $\varepsilon$DP-DR estimates were almost equal to the true value 3. This result shows that the bias correction by $1-\hat{\varepsilon}$ worked well in our experiments.

As in many other outlier-resistant statistical methods, parameter tuning is challenging. We suggest a possible policy on this issue based on the solution paths of the proposed estimators, which is provided in the Appendix. Looking at the paths, the influence of outliers decreased as $\gamma$ increased, and the paths became stable around the true value after reaching a certain $\gamma$. Thus, we suggest using the smallest $\gamma$ for which the estimate is stable.

\begin{spacing}{1}
% latex table generated in R 4.0.3 by xtable 1.8-4 package
% Thu May 06 09:34:29 2021
\begin{landscape}
\begin{table}[ht]
\centering
\small
\begin{tabular}{cccccccccc}
  \hline
 & PS/OR & $\varepsilon$ & $\gamma=0.0$ & 0.1 & 0.2 & 0.5 & 1.0 & 1.5 & 2.0 \\ 
  \hline
DP-IPW & T/- & 0.00 & 3.004 (0.22) & 2.998 (0.22) & 2.994 (0.22) & 2.986 (0.23) & 2.980 (0.26) & 2.974 (0.30) & 2.970 (0.34) \\ 
   &  & 0.05 & 3.749 (0.59) & 3.030 (0.27) & 2.999 (0.26) & 2.987 (0.25) & 2.978 (0.27) & 2.970 (0.30) & 2.963 (0.33) \\ 
   &  & 0.10 & 4.493 (0.78) & 3.142 (0.51) & 3.015 (0.32) & 2.989 (0.27) & 2.977 (0.27) & 2.969 (0.30) & 2.963 (0.33) \\ 
   &  & 0.20 & 5.983 (1.02) & 4.492 (1.70) & 3.536 (1.39) & 3.052 (0.64) & 2.990 (0.41) & 2.978 (0.39) & 2.971 (0.40)\vspace{2mm}\\ 
  DP-DR & T/T & 0.00 & 2.999 (0.18) & 2.998 (0.18) & 2.997 (0.19) & 2.996 (0.20) & 2.992 (0.24) & 2.989 (0.28) & 2.985 (0.31) \\ 
   &  & 0.05 & 3.745 (0.60) & 3.029 (0.30) & 3.002 (0.27) & 2.997 (0.29) & 2.991 (0.29) & 2.985 (0.31) & 2.980 (0.34) \\ 
   &  & 0.10 & 4.489 (0.79) & 3.140 (0.55) & 3.017 (0.36) & 3.000 (0.33) & 2.992 (0.31) & 2.986 (0.32) & 2.981 (0.33) \\ 
   &  & 0.20 & 5.979 (1.04) & 4.465 (1.72) & 3.532 (1.41) & 3.060 (0.69) & 3.009 (0.52) & 2.999 (0.51) & 2.994 (0.51)\vspace{2mm}\\ 
   & T/F & 0.00 & 3.004 (0.24) & 2.998 (0.24) & 2.994 (0.24) & 2.986 (0.25) & 2.979 (0.28) & 2.974 (0.32) & 2.968 (0.36) \\ 
   &  & 0.05 & 3.750 (0.60) & 3.033 (0.31) & 3.001 (0.29) & 2.989 (0.32) & 2.978 (0.33) & 2.970 (0.36) & 2.963 (0.39) \\ 
   &  & 0.10 & 4.494 (0.78) & 3.150 (0.54) & 3.020 (0.37) & 2.992 (0.35) & 2.979 (0.35) & 2.970 (0.37) & 2.963 (0.39) \\ 
   &  & 0.20 & 5.984 (1.03) & 4.490 (1.71) & 3.546 (1.41) & 3.059 (0.71) & 3.001 (0.58) & 2.985 (0.55) & 2.975 (0.54)\vspace{2mm}\\ 
   & F/T & 0.00 & 2.999 (0.18) & 2.999 (0.18) & 2.999 (0.18) & 3.001 (0.20) & 3.005 (0.23) & 3.010 (0.26) & 3.014 (0.29) \\ 
   &  & 0.05 & 3.725 (0.50) & 2.997 (0.19) & 2.976 (0.19) & 2.975 (0.20) & 2.978 (0.23) & 2.982 (0.26) & 2.986 (0.29) \\ 
   &  & 0.10 & 4.451 (0.65) & 3.051 (0.34) & 2.956 (0.21) & 2.950 (0.21) & 2.953 (0.23) & 2.956 (0.26) & 2.960 (0.28) \\ 
   &  & 0.20 & 5.902 (0.86) & 4.326 (1.57) & 3.301 (1.15) & 2.907 (0.35) & 2.895 (0.25) & 2.897 (0.26) & 2.900 (0.28)\vspace{2mm}\\ 
  $\varepsilon$DP-DR & T/T & 0.00 & 2.999 (0.18) & 2.998 (0.18) & 2.997 (0.19) & 2.996 (0.20) & 2.992 (0.24) & 2.989 (0.28) & 2.985 (0.31) \\ 
   &  & 0.05 & 3.745 (0.60) & 3.028 (0.29) & 3.002 (0.27) & 2.997 (0.26) & 2.991 (0.29) & 2.985 (0.31) & 2.980 (0.34) \\ 
   &  & 0.10 & 4.489 (0.78) & 3.138 (0.54) & 3.017 (0.35) & 2.999 (0.30) & 2.991 (0.30) & 2.985 (0.32) & 2.980 (0.33) \\ 
   &  & 0.20 & 5.978 (1.03) & 4.464 (1.72) & 3.531 (1.40) & 3.058 (0.67) & 3.007 (0.51) & 2.998 (0.50) & 2.993 (0.51)\vspace{2mm}\\ 
   & T/F & 0.00 & 3.004 (0.24) & 2.998 (0.24) & 2.994 (0.24) & 2.986 (0.25) & 2.979 (0.28) & 2.974 (0.32) & 2.968 (0.36) \\ 
   &  & 0.05 & 3.750 (0.60) & 3.033 (0.31) & 3.001 (0.29) & 2.989 (0.32) & 2.978 (0.33) & 2.970 (0.36) & 2.963 (0.39) \\ 
   &  & 0.10 & 4.493 (0.78) & 3.149 (0.54) & 3.020 (0.36) & 2.992 (0.34) & 2.978 (0.34) & 2.970 (0.37) & 2.963 (0.39) \\ 
   &  & 0.20 & 5.983 (1.02) & 4.489 (1.71) & 3.543 (1.40) & 3.057 (0.69) & 2.998 (0.55) & 2.984 (0.54) & 2.976 (0.54)\vspace{2mm}\\ 
   & F/T & 0.00 & 2.999 (0.18) & 2.999 (0.18) & 2.999 (0.18) & 3.001 (0.20) & 3.005 (0.23) & 3.010 (0.26) & 3.014 (0.29) \\ 
   &  & 0.05 & 3.746 (0.50) & 3.020 (0.19) & 2.998 (0.19) & 2.998 (0.20) & 3.001 (0.23) & 3.005 (0.26) & 3.009 (0.29) \\ 
   &  & 0.10 & 4.493 (0.66) & 3.108 (0.37) & 3.004 (0.20) & 2.998 (0.21) & 3.001 (0.23) & 3.004 (0.26) & 3.007 (0.28) \\ 
   &  & 0.20 & 5.986 (0.87) & 4.541 (1.58) & 3.486 (1.24) & 3.020 (0.38) & 3.003 (0.24) & 3.005 (0.25) & 3.008 (0.27) \\ 
   \hline
\end{tabular}
\caption{Results of $\gamma$-sensitivity study. Each figure displays the mean (SD) of 10,000 simulations for each setting. In the second column, "T" and "F" denote the correct and the incorrect modeling, respectively.}
\label{tab:exp2}
\end{table}
\end{landscape}
\end{spacing}

%%%%%%%%%%%%%%%%%%%%%%%%%%%%%%%%%%%%%%%%%%%%%%%%%%%%%%%%%%%%%%%%%%%%%%%%%%%%%%%%%%%%%%%%%%%%%%%%%%%%%%%%%%%%%%%%%%%%%%%%%%%%

\section{Real Data Analysis}\label{sec:real}
In this section, we demonstrate the estimation of the ATE on a real dataset. We use the data of the National Health and Nutrition Examination Survey Data I Epidemiologic Follow-up Study (NHEFS). The NHEFS is a national longitudinal study that was performed by U.S. public agencies. We use the processed dataset available online  \citep[][https://www.hsph.harvard.edu/miguel-hernan/causal-inference-book/]{hernan2020causal}. The NHEFS dataset contains 1,566 observations of smokers who were enrolled in the study in 1971--75. By the follow-up visit in 1982, 403 (25.7\%) participants had quit smoking. The study goal was to evaluate the treatment effect of smoking cessation ($T=1$) on weight gain ($Y$). Other than the treatment and outcome, several baseline variables were collected, including sex, age, race, education level, intensity and duration of smoking, physical activity in daily life, recreational exercise, and baseline weight. We used all of them to control for confounding in a similar manner to that of \cite{hernan2020causal}. We included linear and quadratic terms for all continuous covariates (age, intensity and duration of smoking, and baseline weight) and dummy terms for the discrete covariates. The propensity score was estimated by logistic regression, and outcome regression was performed by unnormalized Gaussian modeling (the tuning parameter was set to 0.2). The original dataset does not contain obvious outliers; then, we randomly replaced 10\% observations with outliers drawn from $\mathcal{N}(100,5^2)$. Then, we estimated $\mu^{(1)}$, $\mu^{(0)}$ and the ATE by the same methods in the Monte Carlo simulations. This process was repeated 10,000 times, and we summarized the results in Table \ref{tab:nhefs}. For reference, we estimated every target quantity using the naive IPW/DR using the original data.

For the IPW-type estimators, the median-based methods gave larger estimates of $\mu^{(1)}$ and $\mu^{(0)}$ than those in the case of IPW (no outliers). In particular, $\mu^{(0)}$ was estimated to be much larger. As a result, when using the median-based methods, the ATE was estimated to be smaller than that in the case of IPW (no outliers). By contrast, DP-IPW overestimated $\mu^{(1)}$ with $\gamma=0.05$ and underestimated $\mu^{(1)}$ with $\gamma \ge 0.10$. It overestimated $\mu^{(0)}$ compared to the case of IPW (no outliers), and this tendency was strengthened by increasing $\gamma$. However, because the extent of overestimation of $\mu^{(0)}$ was smaller than that in the case of median-based methods, the estimate of the ATE by DP-IPW was closer to that obtained using IPW (no outliers) than by using the median-based methods. The DR-type estimators showed similar results. The median-based methods overestimated $\mu^{(1)}$ and $\mu^{(0)}$. DP-DR and $\varepsilon$DP-DR underestimated $\mu^{(1)}$ and overestimated $\mu^{(0)}$. The ATE was estimated better by DP-DR and $\varepsilon$DP-DR than by the median-based methods. DP-DR and $\varepsilon$DP-DR had the same tendency of estimation bias and $\gamma$; a larger $\gamma$ value increased the bias.

%%%%%%%%%%%%%%%%%%%%%%%%%%%%%%%%%%%%%%%%%%%%%%%%%%%%%%%%%%%%%%%%%%%%%%%%%%%%%%%%%%%%%%%%%%%%%%%%%%%%%%%%%%%%%%%%%%%%%%%%%%%%

\begin{spacing}{1}
% Table of comparative
% \input{table/table_RMSE_NORM_IDEAL}

% \input{table/table_MEANSD_part}

% \input{table/table_RMSE_exp2_cauchy}

% Table of sensitivity
% \input{table/table_gamma}

% \begin{figure}[htb]
% \centering
% % \includegraphics[width=0.95\textwidth]{fig/exp1_line_1000.png}
% \includegraphics[width=134mm]{fig/exp1_line_100.pdf}
% \caption{Solution paths of the first 100 simulations. The x-axis represents the tuning parameter $\gamma$ and the y-axis, the estimates of $\mu^{(1)}$.}
% \label{fig:exp2}
% \end{figure}

% Table 4
% latex table generated in R 4.0.3 by xtable 1.8-4 package
% Wed Jun 02 10:40:33 2021
\newpage
\begin{table}[ht]
\centering
\begin{tabular}{lccc}
  \hline
 & \multicolumn{3}{c}{{\bf Target Quantities}}  \\ 
 & \multicolumn{1}{c}{$\mu^{(1)}$} & \multicolumn{1}{c}{$\mu^{(0)}$} & \multicolumn{1}{c}{ATE} \\ 
  \hline
IPW (no outliers) & 5.221 (-) & 1.780 (-) & 3.441 (-)\vspace{2mm} \\ 
  IPW & 14.718 (1.57) & 11.607 (0.87) & 3.111 (1.78) \\ 
  median (Firpo) & 5.439 (0.21) & 2.753 (0.10) & 2.686 (0.24) \\ 
  median (Zhang-IPW) & 5.439 (0.21) & 2.753 (0.10) & 2.686 (0.24)\vspace{2mm} \\ 
  DP-IPW ($\gamma=0.05$) & 5.597 (0.30) & 1.851 (0.07) & 3.746 (0.31) \\ 
  DP-IPW ($\gamma=0.10$) & 5.157 (0.15) & 1.819 (0.07) & 3.338 (0.17) \\ 
  DP-IPW ($\gamma=0.20$) & 5.089 (0.15) & 1.875 (0.06) & 3.215 (0.16) \\ 
  DP-IPW ($\gamma=0.50$) & 4.949 (0.15) & 2.007 (0.06) & 2.941 (0.16)\vspace{2mm} \\ 
 \hline\hline  
  DR (no outliers) & 5.136 (-) & 1.772 (-) & 3.364 (-)\vspace{2mm} \\ 
  DR & 14.574 (1.57) & 11.589 (0.90) & 2.985 (1.81) \\ 
  median (Zhang-DR) & 5.352 (0.20) & 2.743 (0.10) & 2.609 (0.22) \\ 
  median (Sued) & 5.353 (0.20) & 2.744 (0.10) & 2.609 (0.23) \\ 
  median (TMLE) & 5.363 (0.21) & 2.739 (0.10) & 2.624 (0.23)\vspace{2mm} \\ 
  DP-DR ($\gamma=0.05$) & 5.478 (0.27) & 1.842 (0.07) & 3.636 (0.28) \\ 
  DP-DR ($\gamma=0.10$) & 5.057 (0.16) & 1.810 (0.07) & 3.248 (0.17) \\ 
  DP-DR ($\gamma=0.20$) & 4.983 (0.16) & 1.865 (0.06) & 3.119 (0.17) \\ 
  DP-DR ($\gamma=0.50$) & 4.834 (0.16) & 1.997 (0.06) & 2.837 (0.17)\vspace{2mm} \\ 
  $\varepsilon$DP-DR ($\gamma=0.05$) & 5.574 (0.29) & 1.851 (0.07) & 3.723 (0.30) \\ 
  $\varepsilon$DP-DR ($\gamma=0.10$) & 5.148 (0.15) & 1.819 (0.07) & 3.330 (0.17) \\ 
  $\varepsilon$DP-DR ($\gamma=0.20$) & 5.080 (0.15) & 1.874 (0.06) & 3.206 (0.17) \\ 
  $\varepsilon$DP-DR ($\gamma=0.50$) & 4.937 (0.15) & 2.007 (0.06) & 2.930 (0.16) \\ 
   \hline
\end{tabular}
\caption{Results of the NHEFS data analysis. Mean and SD are computed on 2,000 bootstrap samples.}
\label{tab:nhefs}
\end{table}
\newpage

\end{spacing}
 
%%%%%%%%%%%%%%%%%%%%%%%%%%%%%%%%%%%%%%%%%%%%%%%%%%%%%%%%%%%%%%%%%%%%%%%%%%%%%%%%%%%%%%%%%%%%%%%%%%%%%%%%%%%%%%%%%%%%%%%%%%%%

\section*{Acknowledgements}

This work was partially supported by JSPS KAKENHI Grant Number 17K00065.
\par

%%%%%%%%%%%%%%%%%%%%%%%%%%%%%%%%%%%%%%%%%%%%%%%%%%%%%%%%%%%%%%%%%%%%%%%%%%%%%%%%%%%%%%%%%%%%%%%%%%%%%%%%%%%%%%%%%%%%%%%%%%%%

\newpage
\bibliographystyle{unsrtnat}  
\bibliography{references}

\clearpage
\appendix
% \section*{Appendix}

% \setcounter{page}{1}
\setcounter{figure}{0}
\setcounter{table}{0}
\setcounter{equation}{0}
\renewcommand{\thesection}{Appendix \Alph{section}}
\renewcommand{\thefigure}{A\arabic{figure}}
\renewcommand{\thetable}{A\arabic{table}}
\renewcommand{\theequation}{A\arabic{equation}}
\renewcommand{\thetheorem}{A\arabic{theorem}}
\renewcommand{\thelemma}{A\arabic{lemma}}

\section{Proofs for Unbiasedness of Estimating Equations}
\subsection{Proof of Theorem 1}
\begin{proof}
\begin{align}
    \mathbb{E}_{g}\left[
        \frac{Th(Y;\mu^{(1)})^{\gamma}}{\pi(X;\alpha^*)}(Y-\mu^{(1)})
    \right]
    =&~ \mathbb{E}_{g}\left[
        \left.\mathbb{E}_{g}\left[
        \frac{Th(Y;\mu^{(1)})^{\gamma}}{\pi(X;\alpha^*)}(Y-\mu^{(1)})
        \right|X\right]
    \right] \nonumber\\
    =&~ \mathbb{E}_{g}\left[
        \frac{P(T=1|X)}{\pi(X;\alpha^*)}
        \left.\mathbb{E}_{g}\left[
            h(Y;\mu^{(1)})^{\gamma}(Y-\mu^{(1)})
        \right|T=1,X\right]
    \right] \nonumber\\
    =&~ \mathbb{E}_{f_1}\left[
            h(Y^{(1)};\mu^{(1)})^{\gamma}(Y^{(1)}-\mu^{(1)})
    \right] \nonumber
\end{align}
The third equality holds from the causal consistency and the conditional unconfoundedness. Since $h(y;\mu^{(1)})$ and $f_1(y)$ are symmetric about $\mu^{(1)}$, this expectation is equal to zero:
\begin{align}
    \mathbb{E}_{f_1}\left[
        h(Y^{(1)};\mu^{(1)})^{\gamma}(Y^{(1)}-\mu^{(1)})
    \right] 
    = \int h(y;\mu^{(1)})^{\gamma}(y-\mu^{(1)})f_1(y) dy 
    = 0 \nonumber.
\end{align}
\end{proof}

\subsection{Proof of Theorem 2}
\begin{proof}
\begin{align}
    &~ \mathbb{E}_{\tilde{g}}\left[
        \frac{Th(Y;\mu^{(1)})^{\gamma}}{\pi(X;\alpha^*)}(Y-\mu^{(1)})
    \right] \nonumber\\
    =&~ \mathbb{E}_{g}\left[
        \left.\mathbb{E}_{\tilde{g}}\left[
            h(Y^{(1)};\mu^{(1)})^{\gamma}(Y^{(1)}-\mu^{(1)})
        \right|X\right]
    \right] \nonumber\\
    =&~ \int\left\{
        (1-\varepsilon_1(x))\int h(y;\mu^{(1)})^{\gamma}(y-\mu^{(1)})g(y|x)dy + \varepsilon_1(x)\int h(y;\mu^{(1)})^{\gamma}(y-\mu^{(1)})\delta_1(y|x)dy\right\}g(x) dx \nonumber\\
    =&~ \int
        (1-\varepsilon_1(x))\int h(y;\mu^{(1)})^{\gamma}(y-\mu^{(1)})g(y|x)dy~g(x)dx + \nu_1(\varepsilon_1) \\
    =&~ -\int
        \varepsilon_1(x)\int h(y;\mu^{(1)})^{\gamma}(y-\mu^{(1)})g(y|x)dy~g(x)dx + \nu_1(\varepsilon_1).
\end{align}
If $\varepsilon_1(x)=\varepsilon_1$, the first term disappears:
\begin{align}
    -\varepsilon_1\iint h(y;\mu^{(1)})^{\gamma}(y-\mu^{(1)})g(y|x)g(x)dydx
    = -\varepsilon_1\int h(y;\mu^{(1)})^{\gamma}(y-\mu^{(1)})f_1(y)dy=0. \nonumber
\end{align}
\end{proof}

\subsection{Proof of Theorem 3}
\begin{proof}
First, we assume that the true PS is given.
\begin{align}
    &~ \mathbb{E}_g\left[
        \frac{Th(Y;\mu^{(1)})^{\gamma}}{\pi(X;\alpha^*)}(Y-\mu^{(1)})
        - \frac{T-\pi(X;\alpha^*)}{\pi(X;\alpha^*)}\mathbb{E}_{\hat{q}}\left[h(Y;\mu^{(1)})^{\gamma}(Y-\mu)|T=1,X\right]
    \right] \nonumber\\
    =&~ \mathbb{E}_g\left[
        \left.\mathbb{E}_g\left[
        \frac{Th(Y;\mu^{(1)})^{\gamma}}{\pi(X;\alpha^*)}(Y-\mu^{(1)})
        - \frac{T-\pi(X;\alpha^*)}{\pi(X;\alpha^*)}\mathbb{E}_{\hat{q}}\left[h(Y;\mu^{(1)})^{\gamma}(Y-\mu)|T=1,X\right]
        \right|X\right]
    \right] \nonumber\\
    =&~ \mathbb{E}_g\left[
        -\left.\mathbb{E}_g\left[
        \frac{T-\pi(X;\alpha^*)}{\pi(X;\alpha^*)}
        \right|X\right]\mathbb{E}_{\hat{q}}\left[h(Y;\mu^{(1)})^{\gamma}(Y-\mu)|T=1,X\right]
    \right] \nonumber\\
    =&~ 0 \nonumber
\end{align} 
Next, we assume that the true OR model is given.
\begin{align}
    &~ \mathbb{E}_g\left[
        \frac{Th(Y;\mu^{(1)})^{\gamma}}{\pi(X;\alpha)}(Y-\mu^{(1)})
        - \frac{T-\pi(X;\alpha)}{\pi(X;\alpha)}\mathbb{E}_{g}\left[h(Y;\mu)^{\gamma}(Y-\mu)|T=1,X\right]
    \right] \nonumber\\
    =&~ \mathbb{E}_g\left[
        \frac{Th(Y;\mu^{(1)})^{\gamma}}{\pi(X;\alpha)}(Y-\mu^{(1)})
        - \frac{T-\pi(X;\alpha)}{\pi(X;\alpha)}\mathbb{E}_{g}\left[h(Y^{(1)};\mu^{(1)})^{\gamma}(Y^{(1)}-\mu^{(1)})|X\right]
    \right] \nonumber\\
    =&~ \mathbb{E}_g\left[
        \left.\mathbb{E}_g\left[\frac{T}{\pi(X;\alpha)}
        h(Y;\mu^{(1)})^{\gamma}(Y-\mu^{(1)})\right|X\right] \right.\nonumber\\
    &~~~~~~~~~~~~~~~~~~~\left.
        - \left.\mathbb{E}_g\left[\frac{T}{\pi(X;\alpha)}-1\right|X\right]\mathbb{E}_g[h(Y^{(1)};\mu^{(1)})^{\gamma}(Y^{(1)}-\mu^{(1)})|X]
    \right] \nonumber\\
    =&~ \mathbb{E}_g\left[
        \frac{P(T=1|X)}{\pi(X;\alpha)}\mathbb{E}_g[
        h(Y^{(1)};\mu^{(1)})^{\gamma}(Y^{(1)}-\mu^{(1)})|X] \right.\nonumber\\
    &~~~~~~~~~~~~~~~~~~~\left.
        - \left(\frac{P(T=1|X)}{\pi(X;\alpha)}-1\right)\mathbb{E}_g[h(Y^{(1)};\mu^{(1)})^{\gamma}(Y^{(1)}-\mu^{(1)})|X]
    \right] \nonumber\\
    =&~ \mathbb{E}_g\left[
        \mathbb{E}_g[h(Y^{(1)};\mu^{(1)})^{\gamma}(Y^{(1)}-\mu^{(1)})|X]
    \right] \nonumber\\
    =&~ 0 \nonumber
\end{align}
Thus, the DP-DR estimating equation has double robustness under no contamination.
\end{proof}

\subsection{Proof of Theorem 4}
\begin{proof}
If the true PS model is given, the DP-DR estimating equation yields
\begin{align}
    &~ \mathbb{E}_{\tilde{g}}\left[
        \frac{Th(Y;\mu^{(1)})^{\gamma}}{\pi(X;\alpha^*)}(Y-\mu^{(1)})
        - \frac{T-\pi(X;\alpha^*)}{\pi(X;\alpha^*)}\mathbb{E}_{\hat{q}}\left[h(Y;\mu^{(1)})^{\gamma}(Y-\mu)|T=1,X\right]
    \right] \nonumber \\
    =&~ \mathbb{E}_g\left[
        \left.\mathbb{E}_{\tilde{g}}\left[
        \frac{Th(Y;\mu^{(1)})^{\gamma}}{\pi(X;\alpha^*)}(Y-\mu^{(1)})
        - \frac{T-\pi(X;\alpha^*)}{\pi(X;\alpha^*)}\mathbb{E}_{\hat{q}}\left[h(Y;\mu^{(1)})^{\gamma}(Y-\mu)|T=1,X\right]
        \right|X\right]
    \right] \nonumber\\
    =&~ \mathbb{E}_g\left[
        \left.\mathbb{E}_{\tilde{g}}\left[
        \frac{Th(Y;\mu^{(1)})^{\gamma}}{\pi(X;\alpha^*)}(Y-\mu^{(1)})
        \right|X\right]
    \right] \nonumber\\
    &~~~~~ - \underbrace{\mathbb{E}_g\left[
        \left.\mathbb{E}_g\left[
        \frac{T-\pi(X;\alpha^*)}{\pi(X;\alpha^*)}
        \right|X\right]\mathbb{E}_{\hat{q}}\left[h(Y;\mu^{(1)})^{\gamma}(Y-\mu)|T=1,X\right]
    \right]}_{=0} \nonumber\\
    =&~ -\int
        \varepsilon_1(x)\int h(y;\mu^{(1)})^{\gamma}(y-\mu^{(1)})g(y|x)dy~g(x)dx + \nu_1(\varepsilon_1).
\end{align} 
If the contamination ratio is independent of $X$, it holds that 
\begin{align}
    -\varepsilon_1\iint h(y;\mu^{(1)})^{\gamma}(y-\mu^{(1)})g(y|x)dy~g(x)dx = 0 \nonumber,
\end{align}
which is the same result as the DP-IPW estimating equation.

If the true OR model is given, the DP-DR estimating equation yields
\begin{align}
    &~ \mathbb{E}_{\tilde{g}}\left[
        \frac{Th(Y;\mu^{(1)})^{\gamma}}{\pi(X;\alpha)}(Y-\mu^{(1)})
        - \frac{T-\pi(X;\alpha)}{\pi(X;\alpha)}\mathbb{E}_{g}\left[h(Y;\mu^{(1)})^{\gamma}(Y-\mu)|T=1,X\right]
    \right] \nonumber\\
    =&~ \mathbb{E}_{\tilde{g}}\left[
        \frac{Th(Y;\mu^{(1)})^{\gamma}}{\pi(X;\alpha)}(Y-\mu^{(1)})
        - \frac{T-\pi(X;\alpha)}{\pi(X;\alpha)}\mathbb{E}_{g}\left[h(Y^{(1)};\mu^{(1)})^{\gamma}(Y^{(1)}-\mu)|X\right]
    \right] \nonumber\\
    =&~ \mathbb{E}_g\left[
        \frac{P(T=1|X)}{\pi(X;\alpha)}\left(
            (1-\varepsilon_1(X))\mathbb{E}_{g}[h(Y^{(1)};\mu^{(1)})^{\gamma}(Y^{(1)}-\mu^{(1)})|X]+\varepsilon_1(X)\mathbb{E}_{\delta}[h(Y;\mu^{(1)})^{\gamma}(Y-\mu^{(1)})|X]
        \right)\right.\nonumber\\
        &~~~~~~ - \left.\left(\frac{P(T=1|X)}{\pi(X;\alpha)}-1\right)\mathbb{E}_g[h(Y^{(1)};\mu^{(1)})^{\gamma}(Y^{(1)}-\mu^{(1)})|X]
    \right] \nonumber\\
    =&~ \mathbb{E}_g\left[
        -\varepsilon_1(X)\frac{P(T=1|X)}{\pi(X;\alpha)}\mathbb{E}_g[h(Y^{(1)};\mu^{(1)})^{\gamma}(Y^{(1)}-\mu^{(1)})|X]
    \right] + \nu_1(\varepsilon_1(\cdot)P(T=1|\cdot)/\pi(\cdot;\alpha)).
\end{align}
When the contamination ratio is independent of $X$, the first term becomes
\begin{align}
    -\varepsilon_1\mathbb{E}_g\left[
        \frac{P(T=1|X)}{\pi(X;\alpha)}\mathbb{E}_g[h(Y^{(1)};\mu^{(1)})^{\gamma}(Y^{(1)}-\mu^{(1)})|X]
    \right]
\end{align}
Thus, we have Theorem 4.
\end{proof}

\section{Derivation of Influence functions in Section 4}
\subsection{DP-IPW}
Let $\tilde{\mu}_n^{(1)}$ denote the root of the DP-IPW estimating equation under contamination.
\begin{align}
    0 =&~ \left.\frac{\partial}{\partial\varepsilon_1(X_i)}\left\{\frac{1}{n}\sum_{i=1}^n\left.\mathbb{E}_{\tilde{g}}\left[\frac{Th(Y;\mu_n^{(1)})^{\gamma}}{\pi(X_i;\alpha^*)}(Y-\tilde{\mu}_n^{(1)})\right|X_i\right]\right\}\right|_{\varepsilon_1(X_i)=0}\nonumber\\
    =&~ \left.\frac{\partial}{\partial\varepsilon_1(X_i)}\iint\frac{t}{\pi(X_i;\alpha^*)}h(y;\mu_n^{(1)})^{\gamma}(y-\tilde{\mu}_n^{(1)})\{(1-\varepsilon_1(X_i))g(y|X_i)+\varepsilon_1(X_i)\delta_{y_0}(y)\}g(t|X_i)dydt\right|_{\varepsilon_1(X_i)=0} \nonumber\\
    =&~ \iint\left.\frac{\partial\psi}{\partial\mu}\right|_{\mu=\mu^{(1)}_n}g(y|X_i)g(t|X_i)dydt\cdot IF(y_0) \nonumber\\
    &~~~~~ + \iint\frac{t}{\pi(X_i;\alpha^*)}h(y;\mu_n^{(1)})^{\gamma}(y-\mu_n^{(1)})\delta_{y_0}(y)g(t|X_i)dydt \nonumber\\
    =&~ \left.\mathbb{E}_{g}\left[\left.\frac{\partial\psi}{\partial\mu}\right|_{\mu=\mu^{(1)}_n}\right|X_i\right]\cdot IF_{DP-IPW}(y_0) + h(y_0;\mu_n^{(1)})^{\gamma}(y_0-\mu^{(1)}_n)\nonumber
\end{align}
If $\left.\mathbb{E}_{g}\left[\left.\frac{\partial\psi}{\partial\mu}\right|_{\mu=\mu^{(1)}_n}\right|X_i\right]$ is invertible, we obtain the IF as
\begin{align}
    -\left.\mathbb{E}_{g}\left[\left.\frac{\partial\psi}{\partial\mu}\right|_{\mu=\mu^{(1)}_n}\right|X_i\right]^{-1}h(y_0;\mu_n^{(1)})^{\gamma}(y_0-\mu^{(1)}_n).
\end{align}

\subsection{DP-DR}
\begin{align}
    0 =&~ \frac{\partial}{\partial\varepsilon_1(X_i)}\left\{\frac{1}{n}\sum_{i=1}^n\left.\mathbb{E}_{\tilde{g}}\left[\frac{Th(Y;\tilde{\mu}_n^{(1)})^{\gamma}}{\pi(X_i;\alpha^*)}(Y-\tilde{\mu}_n^{(1)})\right|X_i\right]\right. \nonumber\\
    &~~~~~~~~~~~~~ \left.\left.- \mathbb{E}\left[\left.\frac{T-\pi(X_i;\alpha^*)}{\pi(X_i;\alpha^*)}\right|X_i\right]\mathbb{E}_{\hat{q}}[h(Y^{(1)};\tilde{\mu}_n^{(1)})^{\gamma}(Y^{(1)}-\tilde{\mu}_n^{(1)})|X_i]\right\}\right|_{\varepsilon_1(X_i)=0} \nonumber\\
    =&~ \frac{\partial}{\partial\varepsilon_1(X_i)}\left\{
    \iint\frac{t}{\pi(X_i;\alpha^*)}h(y;\tilde{\mu}_n^{(1)})^{\gamma}(y-\tilde{\mu}_n^{(1)})\{(1-\varepsilon_1(X_i))g(y|X_i)+\varepsilon_1(X_i)\delta_{y_0}(y)\}g(t|X_i) \right. \nonumber\\
    &~~~~~~~~~~~~~ \left.\left.- \frac{t-\pi(X_i;\alpha^*)}{\pi(X_i;\alpha^*)}\mathbb{E}_{\hat{q}}[h(Y^{(1)};\tilde{\mu}_n^{(1)})^{\gamma}(Y^{(1)}-\tilde{\mu}_n^{(1)})|X_i]g(t|X_i)dydt\right\}\right|_{\varepsilon_1(X_i)=0} \nonumber\\
    =&~ \left.\mathbb{E}_{g}\left[\left.\frac{\partial\psi}{\partial\mu}\right|_{\mu=\mu^{(1)}_n}\right|X_i\right]\cdot IF(y_0) + \frac{P(T=1|X_i)}{\pi(X_i;\alpha^*)}h(y_0;\mu_n^{(1)})^{\gamma}(y_0-\mu_n^{(1)})\nonumber\\
    &~~~~~~~ - \frac{P(T=1|X_i)-\pi(X_i;\alpha^*)}{\pi(X_i;\alpha^*)}\mathbb{E}_{\hat{q}}[h(Y^{(1)};\mu_n^{(1)})^{\gamma}(Y^{(1)}-\mu_n^{(1)})|X_i]\nonumber
\end{align}
Then, we obtain the IF as
\begin{align}
    -\mathbb{E}_{g}\left[\left.\left.\frac{\partial\psi}{\partial\mu}\right|_{\mu=\mu^{(1)}_n}\right|X_i\right]^{-1}&\left\{\frac{P(T=1|X_i)}{\pi(X_i;\alpha)}h(y_{0};\mu_n^{(1)})^{\gamma}(y_{0}-\mu^{(1)}_n)\right.\nonumber\\
    &\left.-\frac{P(T=1|X_i)-\pi(X_i;\alpha)}{\pi(X_i;\alpha)}\mathbb{E}_{\hat{q}}[h(Y^{(1)};\mu_n^{(1)})^{\gamma}(Y^{(1)}-\mu_n^{(1)})|X_i]\right\}.
\end{align}

\subsection{$\varepsilon$DP-DR}
Suppose the expected contamination ratio is correctly specified as $\overline{\varepsilon}_1= (1-\sum\varepsilon_1(X_i)/n)$.
\begin{align}
    0 =&~ \frac{\partial}{\partial\varepsilon_1(X_i)}\left\{\frac{1}{n}\sum_{i=1}^n\left.\mathbb{E}_{\tilde{g}}\left[\frac{Th(Y;\tilde{\mu}_n^{(1)})^{\gamma}}{\pi(X_i;\alpha^*)}(Y-\tilde{\mu}_n^{(1)})\right|X_i\right]\right. \nonumber\\
    &~~~~~~~ \left.\left.- \left(1-\frac{1}{n}\sum_{i=1}^n\varepsilon_1(X_i)\right) \mathbb{E}\left[\left.\frac{T-\pi(X_i;\alpha^*)}{\pi(X_i;\alpha^*)}\right|X_i\right]\mathbb{E}_{\hat{q}}[h(Y^{(1)};\tilde{\mu}_n^{(1)})^{\gamma}(Y^{(1)}-\tilde{\mu}_n^{(1)})|X_i]\right\}\right|_{\varepsilon_1(X_i)=0}\nonumber \\
    =&~ \frac{\partial}{\partial\varepsilon_1(X_i)}\left\{
    \iint\frac{t}{\pi(X_i;\alpha^*)}h(y;\tilde{\mu}_n^{(1)})^{\gamma}(y-\tilde{\mu}_n^{(1)})\{(1-\varepsilon_1(X_i))g(y|X_i)+\varepsilon_1(X_i)\delta_{y_0}(y)\}g(t|X_i) \right. \nonumber\\
    &~~~~~~~~~ \left.\left.- \left(1-\frac{1}{n}\sum_{i=1}^n\varepsilon_1(X_i)\right)\frac{t-\pi(X_i;\alpha^*)}{\pi(X_i;\alpha^*)}\mathbb{E}_{\hat{q}}[h(Y^{(1)};\tilde{\mu}_n^{(1)})^{\gamma}(Y^{(1)}-\tilde{\mu}_n^{(1)})|X_i]g(t|X_i)dydt\right\}\right|_{\varepsilon_1(X_i)=0} \nonumber\\
    =&~ \left.\mathbb{E}_{g}\left[\left.\frac{\partial\psi}{\partial\mu}\right|_{\mu=\mu^{(1)}_n}\right|X_i\right]\cdot IF(y_0) + \frac{P(T=1|X_i)}{\pi(X_i;\alpha^*)}h(y_0;\mu_n^{(1)})^{\gamma}(y_0-\mu_n^{(1)})\nonumber\\
    &~~~~~~~ - \frac{n-1}{n}\frac{P(T=1|X_i)-\pi(X_i;\alpha^*)}{\pi(X_i;\alpha^*)}\mathbb{E}_{\hat{q}}[h(Y^{(1)};\mu_n^{(1)})^{\gamma}(Y^{(1)}-\mu_n^{(1)})|X_i]\nonumber
\end{align}
Thus, we obtain (4.25).

\section{Influence Functions Under Homogeneous Contamination}
Under homogeneous contamination, we can apply the ordinary IF analysis. By differentiating the estimating equations with respect to $\varepsilon_1$ at $\varepsilon_1=0$, we obtain the following results.
\subsection{DP-IPW}
\begin{align}
    0 =&~ \left.\frac{\partial}{\partial\varepsilon_1}\mathbb{E}_{\tilde{g}}\left[
        \frac{Th(Y;\tilde{\mu}^{(1)})^{\gamma}}{\pi(X;\alpha^*)}(Y-\tilde{\mu}^{(1)})
    \right]\right|_{\varepsilon_1=0} \nonumber\\
    0 =&~ \left.\frac{\partial}{\partial\varepsilon_1}\iiint 
        \frac{th(y;\tilde{\mu}^{(1)})^{\gamma}}{\pi(x;\alpha^*)}(y-\tilde{\mu}^{(1)})
    \{(1-\varepsilon_1)g(y|t,x)+\varepsilon_1\delta_{y_0}(y|x)\}g(t|x)g(x)dydtdx\right|_{\varepsilon_1=0} \nonumber\\
    =&~ \mathbb{E}_{g}\left[\left.\frac{\partial\psi}{\partial\varepsilon_1}\right|_{\mu=\mu^{(1)}}\right]\cdot IF(y_0) + \iiint 
        \frac{th(y;\mu^{(1)})^{\gamma}}{\pi(x;\alpha^*)}(y-\mu^{(1)})
    \delta_{y_0}(y|x)g(t|x)g(x)dydtdx \nonumber\\
    IF(y_0) =&~ \mathbb{E}_{g}\left[\left.\frac{\partial\psi}{\partial\varepsilon_1}\right|_{\mu=\mu^{(1)}}\right]^{-1}h(y_0;\mu^{(1)})^{\gamma}(y_0-\mu^{(1)})
\end{align}
Thus, DP-IPW has a redescending property under homogeneous contamination.

\subsection{DP-DR}
\begin{align}
    0 =&~ \left.\frac{\partial}{\partial\varepsilon_1}
    \mathbb{E}_{\tilde{g}}\left[
        \frac{Th(Y;\tilde{\mu}^{(1)})^{\gamma}}{\pi(X;\alpha)}(Y-\tilde{\mu}^{(1)}) - \frac{T-\pi(X;\alpha)}{\pi(X;\alpha)}\mathbb{E}_{\hat{q}}[h(Y^{(1)};\tilde{\mu}^{(1)})^{\gamma}(Y^{(1)}-\tilde{\mu}^{(1)})|X]
    \right]\right|_{\varepsilon_1=0}\nonumber \\
    =&~ \frac{\partial}{\partial\varepsilon_1}\iiint \left(
        \frac{th(y;\tilde{\mu}^{(1)})^{\gamma}}{\pi(x;\alpha)}(y-\tilde{\mu}^{(1)}) - \frac{t-\pi(x;\alpha)}{\pi(x;\alpha)}\mathbb{E}_{\hat{q}}[h(Y^{(1)};\tilde{\mu}^{(1)})^{\gamma}(Y^{(1)}-\tilde{\mu}^{(1)})|x]\right)\nonumber\\
    &~~~~~~~~ \times\left.\{(1-\varepsilon_1)g(y|t,x)+\varepsilon_1\delta_{y_0}(y|x)\}g(t|x)g(x)dydtdx\right|_{\varepsilon_1=0} \nonumber \\
    =&~ \mathbb{E}_{g}\left[\left.\frac{\partial\psi}{\partial\varepsilon_1}\right|_{\mu=\mu^{(1)}}\right]\cdot IF(y_0) + \iiint \left(
        \frac{th(y;\mu^{(1)})^{\gamma}}{\pi(x;\alpha)}(y-\mu^{(1)})\right.\nonumber\\
    &~~~~~~ - \left.\frac{t-\pi(x;\alpha)}{\pi(x;\alpha)}\mathbb{E}_{\hat{q}}[h(Y^{(1)};\mu^{(1)})^{\gamma}(Y^{(1)}-\mu^{(1)})|x]\right)\delta_{y_0}(y|x)g(t|x)g(x)dydtdx \nonumber
\end{align}
\begin{align}
    IF(y_0) =&~ \mathbb{E}_{g}\left[\left.\frac{\partial\psi}{\partial\varepsilon_1}\right|_{\mu=\mu^{(1)}}\right]^{-1}\iint \left(
        \frac{th(y_0;\mu^{(1)})^{\gamma}}{\pi(x;\alpha)}(y_0-\mu^{(1)}) \right.\nonumber\\
    &~~~~~~~~ \left.- \frac{t-\pi(x;\alpha)}{\pi(x;\alpha)}\mathbb{E}_{\hat{q}}[h(Y^{(1)};\mu^{(1)})^{\gamma}(Y^{(1)}-\mu^{(1)})|x]\right)g(t|x)g(x)dtdx
\end{align}
If the true PS model is given, this IF reduces to 
\begin{align}
    IF(y_0) =&~ \mathbb{E}_{g}\left[\left.\frac{\partial\psi}{\partial\varepsilon_1}\right|_{\mu=\mu^{(1)}}\right]^{-1}h(y_0;\mu^{(1)})^{\gamma}(y_0-\mu^{(1)})
\end{align}
If the true OR model is given, this IF reduces to
\begin{align}
    IF(y_0) =&~ \mathbb{E}_{g}\left[\left.\frac{\partial\psi}{\partial\varepsilon_1}\right|_{\mu=\mu^{(1)}}\right]^{-1}\int\frac{P(T=1|x)}{\pi(x;\alpha)}h(y_0;\mu^{(1)})^{\gamma}(y_0-\mu^{(1)}) \nonumber\\
    &~~~~~~~~ -  \frac{P(T=1|x)-\pi(x;\alpha)}{\pi(x;\alpha)}\mathbb{E}_{g}[h(Y;\mu^{(1)})^{\gamma}(Y^{(1)}-\mu^{(1)})|x]g(x)dx
\end{align}
Thus, DP-DR has a redescending property under homogeneous contamination in the PS-correct case.

\subsection{$\varepsilon$DP-DR}
\begin{align}
    0 =&~ \left.\frac{\partial}{\partial\varepsilon_1}
    \mathbb{E}_{\tilde{g}}\left[
        \frac{Th(Y;\tilde{\mu}^{(1)})^{\gamma}}{\pi(X;\alpha)}(Y-\tilde{\mu}^{(1)})\right.\right. \nonumber\\
    &~~~~~~ \left.\left.- \frac{T-\pi(X;\alpha)}{\pi(X;\alpha)}(1-\varepsilon_1)\mathbb{E}_{\hat{q}}[h(Y^{(1)};\tilde{\mu}^{(1)})^{\gamma}(Y^{(1)}-\tilde{\mu}^{(1)})|X]
    \right]\right|_{\varepsilon_1=0}\nonumber \\
    =&~ \mathbb{E}_{g}\left[\left.\frac{\partial\psi}{\partial\varepsilon_1}\right|_{\mu=\mu^{(1)}}\right]\cdot IF(y_0) \nonumber\\
    &~~~~ + \iiint \left(
        \frac{th(y;\mu^{(1)})^{\gamma}}{\pi(x;\alpha)}(y-\mu^{(1)})-\frac{t-\pi(x;\alpha)}{\pi(x;\alpha)}\mathbb{E}_{\hat{q}}[h(Y^{(1)};\mu^{(1)})^{\gamma}(Y^{(1)}-\mu^{(1)})|x]\right.\nonumber\\
    &~~~~~~~~~~~~ \left.+\frac{t-\pi(x;\alpha)}{\pi(x;\alpha)}\mathbb{E}_{\hat{q}}[h(Y^{(1)};\mu^{(1)})^{\gamma}(Y^{(1)}-\mu^{(1)})|x]\right)\delta_{y_0}(y|x)g(t|x)g(x)dydtdx \nonumber\\
    IF(y_0) =&~ \mathbb{E}_{g}\left[\left.\frac{\partial\psi}{\partial\varepsilon_1}\right|_{\mu=\mu^{(1)}}\right]^{-1}\int\frac{P(T=1|x)}{\pi(x;\alpha)}h(y_0;\mu^{(1)})^{\gamma}(y_0-\mu^{(1)})g(x)dx
\end{align}
Thus, under homogeneous contamination, $\varepsilon$DP-DR has a redescending property in either the PS-correct case or the OR-correct case.

\section{Further Discussion on Asymptotic Properties}
\subsection{Regularity Conditions for Theorem 5}\label{app:reg}
Detailed discussion is available in Chapter 5 of Van der Vaart (2000), for example.
\begin{enumerate}
    \renewcommand{\labelenumi}{(\alph{enumi})}
    \item The function $S(\lambda)$ is twice continuously differentiable with respect to $\lambda$.
    \item There exists a root $\lambda^*$ of $\mathbb{E}_{\tilde{g}}[S(\lambda)]=0$.
    \item $\mathbb{E}_{\tilde{g}}[\|S(\lambda^*)\|^2]<\infty$.
    \item $\mathbb{E}_{\tilde{g}}[\partial S(\lambda^*)/\partial\lambda^T]$ exists and is nonsingular.
    \item The second-order differentials of $S(\lambda)$ with respect to $\mu$ are dominated by a fixed integrable function $h$ in a neighborhood of $\lambda^*$.
\end{enumerate}

\subsection{Proof of Theorem 6}
Under homogeneous contamination, we see that simpler properties hold.
The matrix $\mathbf{J}^{\tilde{g}}(\lambda^*)$ is partitioned as 
\begin{align}
    \mathbf{J}^{\tilde{g}}(\lambda^*)
    =&~ \left(\begin{array}{ccc}
        \mathbb{E}_{\tilde{g}}\left[\frac{\partial}{\partial\mu}\psi_i(\mu^*;\alpha^*,\beta^*)\right] & \mathbb{E}_{\tilde{g}}\left[\frac{\partial}{\partial\alpha^T}\psi_i(\mu^*;\alpha^*,\beta^*)\right] & \mathbb{E}_{\tilde{g}}\left[\frac{\partial}{\partial\beta^T}\psi_i(\mu^*;\alpha^*,\beta^*)\right] \\
        \mathbf{0} & \mathbb{E}_{g}\left[\frac{\partial}{\partial\alpha^T}s^{PS}_i(\alpha^*)\right] & \mathbf{0} \\
        \mathbf{0} & \mathbf{0} & \mathbb{E}_{\tilde{g}}\left[\frac{\partial}{\partial\beta^T}s^{OR}_i(\beta^*)\right]
    \end{array}\right) \nonumber\\
    =&~ \left(\begin{array}{ccc}
        \mathbf{J}^{\tilde{g}}_{11}(\lambda^*) & \mathbf{J}^{\tilde{g}}_{12}(\lambda^*) & \mathbf{J}^{\tilde{g}}_{13}(\lambda^*) \\
        \mathbf{0} & \mathbf{J}^{g}_{22}(\lambda^*) & \mathbf{0} \\
        \mathbf{0} & \mathbf{0} & \mathbf{J}^{\tilde{g}}_{33}(\lambda^*)
    \end{array}\right).\nonumber
\end{align}
If it is nonsingular, the inverse is obtained as 
\begin{align}
    \mathbf{J}^{\tilde{g}}(\lambda^*)^{-1}
    =&~ \left(\begin{array}{ccc}
        \mathbf{J}^{\tilde{g}}_{11}(\lambda^*)^{-1} & -\mathbf{J}^{\tilde{g}}_{11}(\lambda^*)^{-1}\mathbf{J}^{\tilde{g}}_{12}(\lambda^*)\mathbf{J}^{g}_{22}(\lambda^*)^{-1} & -\mathbf{J}^{\tilde{g}}_{11}(\lambda^*)^{-1}\mathbf{J}^{\tilde{g}}_{13}(\lambda^*)\mathbf{J}^{\tilde{g}}_{33}(\lambda^*)^{-1} \\
        \mathbf{0} & \mathbf{J}^{g}_{22}(\lambda^*)^{-1} & \mathbf{0} \\
        \mathbf{0} & \mathbf{0} & \mathbf{J}^{\tilde{g}}_{33}(\lambda^*)^{-1}
    \end{array}\right).\nonumber
\end{align}
Note that $\mathbf{J}^{\tilde{g}}_{11}(\cdot)$ is a scalar value.

Then, Theorem 5 is proved as follows.
\begin{proof}
By Taylor's theorem, the expectation of the estimating equation (5.26) is expressed as
\begin{align}
    0=\mathbb{E}[S_i(\lambda^{*})] = \mathbb{E}[S_i(\lambda^{**})] + \mathbf{J}^{\tilde{g}}(\lambda^\dag)(\lambda^{*}-\lambda^{**}),\nonumber
\end{align}
where $\lambda^\dag$ is an intermediate value between $\lambda^{**}$ and $\lambda^{*}$.
Since $\mathbb{E}[s^{PS}_i(\alpha^*)]=\mathbb{E}[s^{OR}_i(\beta^*)]=0$ and $(\lambda^{*}-\lambda^{**})=(\mu^*-\mu^{(1)},\mathbf{0}^T,\mathbf{0}^T)^T$, only the first element is meaningful:
\begin{align}
    0 = \mathbb{E}[\psi_i(\mu^{(1)};\alpha^*,\beta^*)] + \mathbf{J}^{\tilde{g}}_{11}(\lambda^{\dag})(\mu^{*}-\mu^{(1)}).\nonumber
\end{align}
Then, since $\mathbf{J}^{\tilde{g}}_{11}(\lambda^{\dag})$ is non-zero, the latent bias of $\mu^*$ reduces to 
\begin{align}\label{eq:xxx}
    \mu^* - \mu^{(1)} = -\mathbf{J}^{\tilde{g}}_{11}(\lambda^{\dag})^{-1}\mathbb{E}[\psi_i(\mu^{(1)};\alpha^*,\beta^*)].
\end{align}
From Corollary 1, if either the PS or the OR model is correct, we have
\begin{align}
    \mathbb{E}[\psi_i(\mu^{(1)};\alpha^*,\beta^*)] = \nu_1(\phi).\nonumber
\end{align}
Upon substituting it into \eqref{eq:xxx}, the statement holds.
\end{proof}

\subsection{Further Discussion on Asymptotic Variance}
Considering the structure of the full estimating equation, the asymptotic variance can be expressed in a more explicit form. The discussion about the asymptotic variance is provided in the next section.

The matrix $\mathbf{K}_{\tilde{g}}(\lambda^*)$ is also partitioned as
\begin{align}
    \mathbf{K}^{\tilde{g}}(\lambda^*)
    =&~ \left(\begin{array}{ccc}
        \mathbf{K}^{\tilde{g}}_{11}(\lambda^*) & \mathbf{K}^{\tilde{g}}_{12}(\lambda^*) & \mathbf{K}^{\tilde{g}}_{13}(\lambda^*) \\
        \mathbf{K}^{\tilde{g}~T}_{12}(\lambda^*) & \mathbf{K}^g_{22}(\lambda^*) & \mathbf{K}^{\tilde{g}}_{23}(\lambda^*) \\
        \mathbf{K}^{\tilde{g}~T}_{13}(\lambda^*) & \mathbf{K}^{\tilde{g}~T}_{23}(\lambda^*) & \mathbf{K}^{\tilde{g}}_{33}(\lambda^*)
    \end{array}\right).\nonumber
\end{align}
The asymptotic variance is also affected by outliers. However, under Assumption 1, the asymptotic variance can be approximated by the asymptotic variance under no contamination and contamination ratio $\varepsilon_1$. 

\begin{theorem}\label{thm:avar}
Besides to Assumption 1, assume that $\mathbf{J}^{\delta}_{1m}(\lambda^{**})\approx\mathbf{0}$ and $\mathbf{K}^{\delta}_{1m}(\lambda^{**})\approx\mathbf{0}$ holds for $m=1,2,3$. 
Then, under homogeneous contamination, 
\begin{align}\label{eq:avar}
    \mathbf{V}^{\tilde{g}}(\lambda^*)
        \approx&~ \mathbf{J}^{\check{g}}(\lambda^{**})^{-1}\left(\begin{array}{ccc}
        \frac{1}{(1-\varepsilon_1)}\mathbf{K}^g_{11}(\lambda^{**}) & \mathbf{K}^g_{12}(\lambda^{**}) & \mathbf{K}^{g}_{13}(\lambda^{**}) \\
        \mathbf{K}^g_{12}(\lambda^{**})^T & \mathbf{K}^g_{22}(\lambda^{**}) & \mathbf{K}^{\tilde{g}}_{23}(\lambda^{**}) \\
        \mathbf{K}^{g}_{13}(\lambda^{**})^T & \mathbf{K}^{\tilde{g}}_{23}(\lambda^{**})^T & \mathbf{K}^{\tilde{g}}_{33}(\lambda^{**})
    \end{array}\right)\{\mathbf{J}^{\check{g}}(\lambda^{**})^T\}^{-1},\nonumber
\end{align}
where
\begin{align}
    \mathbf{J}^{\check{g}}(\lambda^{**}) = \left(\begin{array}{ccc}
        \mathbf{J}^{g}_{11}(\lambda^{**}) & \mathbf{J}^{g}_{12}(\lambda^{**}) & \mathbf{J}^{g}_{13}(\lambda^{**}) \\
        \mathbf{0} & \mathbf{J}^{g}_{22}(\lambda^{**}) & \mathbf{0} \\
        \mathbf{0} & \mathbf{0} & \mathbf{J}^{\tilde{g}}_{33}(\lambda^{**})
    \end{array}\right)\nonumber
\end{align}.
\end{theorem}

If both the PS and the OR models are correct, the asymptotic variance of $\hat{\mu}$ has a simpler expression. From a similar discussion to Section 3, the following lemma holds:
\begin{lemma}\label{lem:JJ} If the PS model is correct, $\mathbf{J}^{g}_{13}(\lambda^{**})=\mathbf{0}$. If the OR model is correct, $\mathbf{J}^{g}_{12}(\lambda^{**})=\mathbf{0}$.
\end{lemma}
Using Lemma \ref{lem:JJ}, we can see the asymptotic variance of $\hat\mu$ is simply expressed.

\begin{theorem}
Under the same assumptions of Theorem \ref{thm:avar}, if the PS and the OR models are both correct, then 
\begin{align}
    \mathbf{V}_{\tilde{g}}(\mu^*) \approx \frac{1}{1-\varepsilon_1}\mathbf{J}^{g}_{11}(\lambda^{**})^{-1}\mathbf{K}^{g}_{11}(\lambda^{**})\left\{\mathbf{J}^{g}_{11}(\lambda^{**})^T\right\}^{-1}.\nonumber
\end{align}
\end{theorem}
\begin{proof}
By applying Lemma \ref{lem:JJ} to Theorem \ref{thm:avar}, the statement holds.
\end{proof}
This implies that the $\varepsilon$DP-DR appropriately ignores outliers.

\subsubsection{Proof of Theorem \ref{thm:avar}}
If either the PS or the OR model is correct, we can say $\mu^*\approx\mu^{(1)}$.
Note that the PS model is not related to the contamination distribution $\delta$, and the contamination on the OR model cannot be removed in general.
From the assumptions, 
\begin{align}
    \mathbf{J}^{\tilde{g}}(\lambda^*)
        \approx&~ \mathbf{J}^{\tilde{g}}(\lambda^{**}) \nonumber\\ 
        =&~ \left(\begin{array}{ccc}
        (1-\varepsilon_1)\mathbf{J}^{g}_{11} & (1-\varepsilon_1)\mathbf{J}^{g}_{12} & (1-\varepsilon_1)\mathbf{J}^{g}_{13} \\
        \mathbf{0} & \mathbf{J}^{g}_{22} & \mathbf{0} \\
        \mathbf{0} & \mathbf{0} & \mathbf{J}^{\tilde{g}}_{33}
    \end{array}\right) + \left(\begin{array}{ccc}
        \varepsilon_1\mathbf{J}^{\delta}_{11} & \varepsilon_1\mathbf{J}^{\delta}_{12} & \varepsilon_1\mathbf{J}^{\delta}_{13} \\
        \mathbf{0} & \mathbf{0} & \mathbf{0} \\
        \mathbf{0} & \mathbf{0} & \mathbf{0}
    \end{array}\right) \nonumber\\
    \approx&~\left(\begin{array}{ccc}
        (1-\varepsilon_1)\mathbf{J}^{g}_{11} & (1-\varepsilon_1)\mathbf{J}^{g}_{12} & (1-\varepsilon_1)\mathbf{J}^{g}_{13} \\
        \mathbf{0} & \mathbf{J}^{g}_{22} & \mathbf{0} \\
        \mathbf{0} & \mathbf{0} & \mathbf{J}^{\tilde{g}}_{33}
    \end{array}\right) \nonumber\\
    =&~\left(\begin{array}{ccc}
        1-\varepsilon_1 & \mathbf{0} & \mathbf{0} \\
        \mathbf{0} & \mathbf{I}_{\alpha} & \mathbf{0} \\
        \mathbf{0} & \mathbf{0} & \mathbf{I}_{\beta}
    \end{array}\right)\mathbf{J}^{\check{g}}(\lambda^{**}) \nonumber
\end{align}
\begin{align}
    \mathbf{K}^{\tilde{g}}(\lambda^*) 
        \approx&~ \mathbf{K}^{\tilde{g}}(\lambda^{**}) \nonumber\\
        =&~ \left(\begin{array}{ccc}
        (1-\varepsilon_1)\mathbf{K}^{g}_{11} & (1-\varepsilon_1)\mathbf{K}^{g}_{12} & (1-\varepsilon_1)\mathbf{K}^{g}_{13} \\
        (1-\varepsilon_1)\mathbf{K}^{g~T}_{12} & \mathbf{K}^{g}_{22} & \mathbf{K}^{\tilde{g}}_{23} \\
        (1-\varepsilon_1)\mathbf{K}^{g~T}_{13} & \mathbf{K}^{\tilde{g}~T}_{23} & \mathbf{K}^{\tilde{g}}_{33}
    \end{array}\right) + \left(\begin{array}{ccc}
        \varepsilon_1\mathbf{K}^{\delta}_{11} & \varepsilon_1\mathbf{K}^{\delta}_{12} & \varepsilon_1\mathbf{K}^{\delta T}_{13} \\
        \varepsilon_1\mathbf{K}^{\delta~T}_{12} & \mathbf{0} & \mathbf{0} \\
        \varepsilon_1\mathbf{K}^{\delta~T}_{13} & \mathbf{0} & \mathbf{0}
    \end{array}\right) \nonumber\\
    \approx&~\left(\begin{array}{ccc}
        (1-\varepsilon_1)\mathbf{K}^{g}_{11} & (1-\varepsilon_1)\mathbf{K}^{g}_{12} & (1-\varepsilon_1)\mathbf{K}^{g}_{13} \\
        (1-\varepsilon_1)\mathbf{K}^{g~T}_{12} & \mathbf{K}^{g}_{22} & \mathbf{K}^{\tilde{g}}_{23} \\
        (1-\varepsilon_1)\mathbf{K}^{g~T}_{13} & \mathbf{K}^{\tilde{g}~T}_{23} & \mathbf{K}^{\tilde{g}}_{33}
    \end{array}\right)\nonumber
\end{align}
The input $(\lambda^{**})$ is dropped for notation simplicity.
Thus, we have
\begin{align}
    \mathbf{V}^{\tilde{g}}(\lambda^*)
        \approx&~ \mathbf{J}^{\check{g}}(\lambda^{**})^{-1}\left(\begin{array}{ccc}
        \frac{1}{(1-\varepsilon_1)}\mathbf{K}^g_{11}(\lambda^{**}) & \mathbf{K}^g_{12}(\lambda^{**}) & \mathbf{K}^{g}_{13}(\lambda^{**}) \\
        \mathbf{K}^g_{12}(\lambda^{**})^T & \mathbf{K}^g_{22}(\lambda^{**}) & \mathbf{K}^{\tilde{g}}_{23}(\lambda^{**}) \\
        \mathbf{K}^{g}_{13}(\lambda^{**})^T & \mathbf{K}^{\tilde{g}}_{23}(\lambda^{**})^T & \mathbf{K}^{\tilde{g}}_{33}(\lambda^{**})
    \end{array}\right)\{\mathbf{J}^{\check{g}}(\lambda^{**})^T\}^{-1},\nonumber
\end{align}

The proof is complete.

\subsubsection{Proof of Lemma \ref{lem:JJ}}
\begin{proof}
In the PS correct case, 
\begin{align}
    \mathbf{J}^{g}_{13}(\lambda^{**}) 
        =&~ \mathbb{E}_{g}\left[\frac{\partial}{\partial\beta^T}\left\{
            \frac{Th(Y;\mu^{(1)})^{\gamma}}{\pi(X;\alpha^*)}(Y-\mu^{(1)})
            - \frac{T-\pi(X;\alpha^*)}{\pi(X;\alpha^*)}\mathbb{E}_{q^*}[h(Y,\mu^{(1)})^{\gamma}(Y-\mu^{(1)})|T=1,X]
        \right\}\right] \nonumber\\
        =&~ \mathbb{E}_{g}\left[
            \frac{P(T=1|X)}{\pi(X;\alpha^*)}h(Y^{(1)};\mu^{(1)})^{\gamma}(Y^{(1)}-\mu^{(1)})\right. \nonumber\\ 
        &~~~~~~~ \left.- \frac{P(T=1|X)-\pi(X;\alpha^*)}{\pi(X;\alpha^*)}\frac{\partial}{\partial\beta^T}\mathbb{E}_{q^*}[h(Y,\mu^{(1)})^{\gamma}(Y-\mu^{(1)})|T=1,X]\right]\nonumber\\
        =&~ \mathbb{E}_{g}\left[h(Y^{(1)};\mu^{(1)})^{\gamma}(Y^{(1)}-\mu^{(1)})\right]\nonumber\\
        =&~ 0,\nonumber
\end{align}
where $\mathbb{E}_{q^*}$ denotes the expectation with respect to $q(y|T=1,x;\beta^{*})$.
In the OR correct case, 
\begin{align}
    \mathbf{J}^{g}_{12}(\lambda^{**}) 
        =&~ \mathbb{E}_{g}\left[\frac{\partial}{\partial\alpha^T}\left\{
            \frac{Th(Y;\mu^{(1)})^{\gamma}}{\pi(X;\alpha^*)}(Y-\mu^{(1)})
            - \frac{T-\pi(X;\alpha^*)}{\pi(X;\alpha^*)}\mathbb{E}_{g}\left[
            h(Y^{(1)};\mu^{(1)})^{\gamma}(Y^{(1)}-\mu^{(1)})
        |X\right]
        \right\}\right]\nonumber\\
        =&~ \mathbb{E}_{g}\left[
            \frac{\partial}{\partial\alpha^T}\left(\frac{P(T=1|X)}{\pi(X;\alpha^*)}-\frac{P(T=1|X)-\pi(X;\alpha^*)}{\pi(X;\alpha^*)}\right)\mathbb{E}_{g}[h(Y^{(1)};\mu^{(1)})^{\gamma}(Y^{(1)}-\mu^{(1)})|X]\right] \nonumber\\
        =&~ \mathbb{E}_{g}\left[
            h(Y^{(1)};\mu^{(1)})^{\gamma}(Y^{(1)}-\mu^{(1)})
        \right] \nonumber\\
        =&~ 0. \nonumber
\end{align}
\end{proof}

\section{Details of Numerical Algorithm}
\subsection{General Form}\label{sec:algo_gen}
Because the proposed estimating equations cannot be solved explicitly, we use an iterative algorithm. Various algorithms are available; however, we propose a standard algorithm for M-estimators. The algorithm for the DP-IPW estimator is given by the following updates:
\begin{align}
    \hat\mu^{[a+1]} &= \left\{\sum_{i=1}^{n}w^{[a]}_iY_i\right\}\left\{\sum_{i=1}^{n}w^{[a]}_i\right\}^{-1},\\
    w^{[a+1]}_i &= \frac{T_ih(Y_i;\hat\mu^{[a+1]})^{\gamma}}{\pi(X_i;\hat{\alpha})}~~~~\text{for all}~~i.\label{eq:ipwupd2}
\end{align}
We recommend obtaining the initial values $(\mu^{[0]}, w_i^{[0]})$ in an outlier-resistant manner. For example, $\mu^{[0]}$ can be obtained by the IPW median (Firpo, 2007; Zhang et al. 2012), and $w_i^{[0]}$ is obtained using \eqref{eq:ipwupd2}. If the weighting density is indexed by other parameters, it must be estimated in advance or be updated simultaneously to $\mu$ and $w$. In the next section, we present an algorithm in which we assume that $h$ is Gaussian.

The ($\varepsilon$)DP-DR estimator is obtained in a similar manner. Let $h(\cdot;\mu^{(1)})$ be fixed and solve the estimating equation of $\varepsilon$DP-DR with respect to $\mu$:
\begin{align}
    \mu =&~ \frac{1}{n}\sum_{i=1}^n \left\{
        \frac{T_ih(Y_i;\mu^{(1)})^{\gamma}Y_i}{\pi(X_i;\hat{\alpha})}-\frac{T_i-\pi(X_i;\hat{\alpha})}{\pi(X_i;\hat{\alpha})}(1-\hat\varepsilon)m_{1,\mu^{(1)}}(X_i;\hat{\beta})
    \right\}\\
    &~~~~~~~\times \left\{
        \frac{T_ih(Y_i;\mu^{(1)})^{\gamma}}{\pi(X_i;\hat{\alpha})}-\frac{T_i-\pi(X_i;\hat{\alpha})}{\pi(X_i;\hat{\alpha})}(1-\hat\varepsilon)m_{0,\mu^{(1)}}(X_i;\hat{\beta})\right\}^{-1},
\end{align}
where $m_{0,\mu^{(1)}}(X_i;\hat{\beta})=\mathbb{E}_g[h(Y^{(1)};\mu)^{\gamma}|X]$ and $m_{1,\mu^{(1)}}(X_i;\hat{\beta})=\mathbb{E}_g[h(Y^{(1)};\mu)^{\gamma}Y^{(1)}|X]$.
Then, the following algorithm is obtained:
\begin{align}
    \hat\mu^{[a+1]} &= \left\{\sum_{i=1}^{n}w^{[a]}_{1,i}Y_i - w_{2,i}\hat{m}_{1,\mu^{[a]}}(X_i;\hat{\beta})\right\}\left\{\sum_{i=1}^{n}w^{[a]}_{1,i} - w_{2,i}\hat{m}_{0,\mu^{[a]}}(X_i;\hat{\beta})\right\}^{-1}, \\
    w^{[a+1]}_{1,i} &= \frac{T_ih(Y_i;\hat\mu^{[a+1]})^{\gamma}}{\pi(X_i;\hat{\alpha})}, ~~~
    w_{2,i} = \frac{T_i-\pi(X_i;\hat{\alpha})}{\pi(X_i;\hat{\alpha})}(1-\hat\varepsilon_1)~~~~\text{for all}~~i.
\end{align}
Note that it is not necessary to update $w_{2,i}$ once it is computed.
The initial values should be obtained in an outlier-resistant manner, as in DP-IPW.
Recall that $\hat{m}_{1,\mu}$ and $\hat{m}_{0,\mu}$ are the estimates for the conditional expectation $\mathbb{E}_g[h(Y^{(1)};\mu)^{\gamma}Y^{(1)}|X]$ and $\mathbb{E}_g[h(Y^{(1)};\mu)^{\gamma}|X]$ given $\mu$. These updates can be obtained from the estimated conditional density $q(y|X;\hat\beta)$ through Monte-Carlo approximation (Hoshino, 2007) or direct calculations.

\subsection{Gaussian Weight}\label{sec:algo_gauss}
When the weighting density $h$ is assumed to be Gaussian, some value must be assigned to the standard deviation $\sigma$. Under contamination, we suggest that $\sigma$ is estimated in an outlier-resistant manner, such as by using the normalized median absolute deviation (MADN). MADN is an unbiased estimator for the standard deviation of a Gaussian random variable. For DP-IPW, we can obtain $\sigma$ by the following updating formula:
\begin{align}
    \hat\sigma^{[a+1]} &= \text{IPW-MADN}(\{Y_i\}_{i=1}^n,\hat\mu^{[a+1]}).
\end{align}
The IPW-MADN is defined as 
\begin{align}
    \text{IPW-MADN}(\{Y_i\}_{i=1}^n,\mu)
        =&~ 1.483\cdot\text{IPW-median}\left(\left\{|Y_i - \mu|\right\}_{i=1}^n\right),
\end{align}
where $1.483$ is a normalization constant.

% The updates when the weight are Gaussian is expressed as follows:
% \begin{align}
%     \hat\mu^{[a+1]} &= \left\{\sum_{i=1}^{n}w^{[a]}_iY_i\right\}\left\{\sum_{i=1}^{n}w^{[a]}_i\right\}^{-1}. \\
%     \hat\sigma^{[a+1]} &= \text{IPW-MADN}(\{Y_i\}_{i=1}^n,\hat\mu^{[a+1]}).\\
%     w^{[a+1]}_i &= \frac{T_ih(Y_i;\theta(\hat\mu^{[a+1]},\hat\sigma^{[a+1]}))^{\gamma}}{\pi(X_i;\hat{\alpha})}~~~~\text{for all}~~i.
% \end{align}
% The IPW-MADN is defined as 
% \begin{align}
%     \text{IPW-MADN}(\{Y_i\}_{i=1}^n,\mu)
%         =&~ 1.483\cdot\text{IPW-median}\left(\left\{|Y_i - \mu|\right\}_{i=1}^n\right),
% \end{align}
% where $1.483$ is a normalization constant.

Similarly, $\sigma$ of the ($\varepsilon$)DP-DR estimator under Gaussian weight is obtained by 
\begin{align}
    \hat\sigma^{[a+1]} &= \text{DR-MADN}(\{Y_i\}_{i=1}^n,\hat\mu^{[a+1]}).
\end{align}
The DR-MADN is obtained by using the DR-median, which is discussed in Section 7.
\begin{align}
    \text{DR-MADN}(\{Y_i\}_{i=1}^n,\mu)
        =&~ 1.483\cdot\text{DR-median}\left(\left\{|Y_i - \mu|\right\}_{i=1}^n\right).
\end{align}

% \begin{align}
%     \hat\mu^{[a+1]} &= \left\{\sum_{i=1}^{n}w^{[a]}_{1,i}Y_i - w_{2,i}\hat{m}_{1,\mu^{[a]}}(X_i;\hat{\beta})\right\}\left\{\sum_{i=1}^{n}w^{[a]}_{1,i} - w_{2,i}\hat{m}_{0,\mu^{[a]}}(X_i;\hat{\beta})\right\}^{-1}, \\
%     \hat\sigma^{[a+1]} &= \text{DR-MADN}(\{Y_i\}_{i=1}^n,\hat\mu^{[a+1]}),\\
%     w^{[a+1]}_{1,i} &= \frac{T_ih(Y_i;\hat\mu^{[a+1]},\hat\sigma^{[a+1]})^{\gamma}}{\pi(X_i;\hat{\alpha})}~~~~\text{for all}~~i, \\
%     w_{2,i} &= \frac{T_i-\pi(X_i;\hat{\alpha})}{\pi(X_i;\hat{\alpha})}(1-\hat\varepsilon_1)~~~~\text{for all}~~i.
% \end{align}
% The DR-MADN is obtained by using the DR-median \citep{zhang2012causal,diaz2017efficient,sued2020robust}
% \begin{align}
%     \text{DR-MADN}(\{Y_i\}_{i=1}^n,\mu)
%         =&~ 1.483\cdot\text{DR-median}\left(\left\{|Y_i - \mu|\right\}_{i=1}^n\right).
% \end{align}

Further, the updates of $\hat{m}_{1,\mu}$ and $\hat{m}_{0,\mu}$ are expressed explicitly when $q(y|X;\hat\beta)$ is assumed to be the conditional Gaussian distribution given $X$. Let $u(X)=\mathbb{E}_q[Y|X]$ and $v^2(X)=\mathrm{Var}_q[Y|X]$. Then, we obtain
\begin{align}
    m_{0,\mu^{[a]}}(X)  
        &= (2\pi)^{-\frac{\gamma}{2}}\frac{({\sigma^{[a]}}^2)^{\frac{1-\gamma}{2}}}{\sqrt{{\sigma^{[a]}}^2+\gamma v^2(X)}}\cdot\exp\left\{-\frac{\gamma(\mu^{[a]}-u(X))}{2({\sigma^{[a]}}^2 + \gamma v^2(X))}\right\}, \\
    m_{1,\mu^{[a]}}(X) 
        &= (2\pi)^{-\frac{\gamma}{2}}\frac{({\sigma^{[a]}}^2)^{\frac{1-\gamma}{2}}}{\sqrt{{\sigma^{[a]}}^2+\gamma v^2(X)}}\cdot \frac{u(X){\sigma^{[a]}}^2 + \gamma\mu^{[a]} v^2(X)}{{\sigma^{[a]}}^2+\gamma v^2(X)}\cdot\exp\left\{-\frac{\gamma(\mu^{[a]}-u(X))}{2({\sigma^{[a]}}^2 + \gamma v^2(X))}\right\}.
\end{align}
Notably, the conditional variance can be easily estimated because many general outlier-resistant methods can be applied for this purpose.

% \newpage
% \section{Results of Simulations on Heavy-tailed Data}
% \input{table/table_RMSE_exp2_cauchy}

\newpage
\section{Solution Paths of $\gamma$-sensitivity Study}
\begin{figure}[htb]
\centering
\includegraphics[width=134mm]{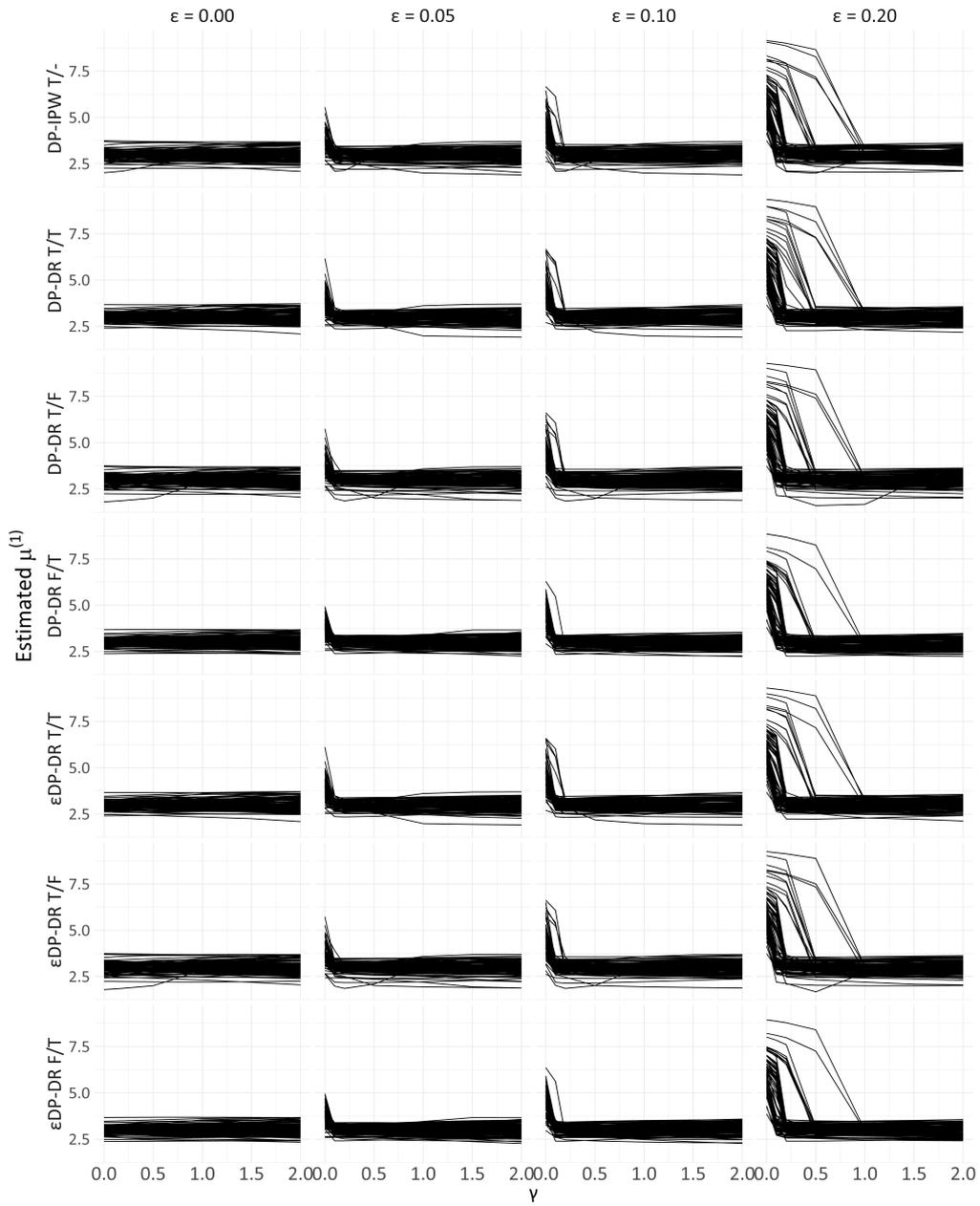}
\caption{Solution paths of the first 100 simulations. The x-axis represents the tuning parameter $\gamma$ and the y-axis, the estimates of $\mu^{(1)}$.}
\end{figure}

\newpage
\section{Remaining Results of Monte-Carlo Simulation}
Remaining results of the Monte-Carlo Simulation are presented in the following pages.

\begin{itemize}
    \item Tables S1 and S2: Gaussian covariates and Gaussian MLE on non-outliers. The RMSE is presented in Tables 2 and 3 in the main text.
    \item Tables S3 to S5: Gaussian covariates and unnormalized Gaussian modeling.
    \item Tables S6 to S8: Uniform covariates and Gaussian MLE on non-outliers.
    \item Tables S9 to S11: Uniform covariates and unnormalized Gaussian modeling.
\end{itemize}

% latex table generated in R 4.0.3 by xtable 1.8-4 package
% Thu Jun 03 09:46:16 2021
\begin{landscape}
  \begin{table}[ht]
  \centering
  \scriptsize
\begin{tabular}{llccccccc}
    \hline
  &  & {\bf No contam.} & \multicolumn{3}{c}{\bf Homogeneous} & \multicolumn{3}{c}{\bf Heterogeneous} \\ 
   &  &  & $\varepsilon=0.05$ & $\varepsilon=0.10$ & $\varepsilon=0.20$ & $\varepsilon=0.05$ & $\varepsilon=0.10$ & $\varepsilon=0.20$ \\
    \hline
IPW(T/-) & Naive & 3.004 (0.22) & 3.749 (0.59) & 4.493 (0.78) & 5.983 (1.02) & 3.766 (0.63) & 4.536 (0.84) & 6.070 (1.08) \\ 
   & median (Firpo) & 2.990 (0.26) & 3.091 (0.28) & 3.205 (0.30) & 3.490 (0.43) & 3.116 (0.28) & 3.259 (0.32) & 3.605 (0.47) \\ 
   & median (Zhang-IPW) & 2.990 (0.26) & 3.091 (0.28) & 3.205 (0.30) & 3.490 (0.43) & 3.116 (0.28) & 3.259 (0.32) & 3.605 (0.47)\vspace{1mm} \\ 
   & DP-IPW ($\gamma=0.1$) & 2.998 (0.22) & 3.030 (0.27) & 3.142 (0.51) & 4.492 (1.70) & 3.056 (0.29) & 3.209 (0.57) & 4.620 (1.74) \\ 
   & DP-IPW ($\gamma=0.5$) & 2.986 (0.23) & 2.987 (0.25) & 2.989 (0.27) & 3.052 (0.64) & 3.011 (0.24) & 3.042 (0.28) & 3.173 (0.70) \\ 
   & DP-IPW ($\gamma=1.0$) & 2.980 (0.26) & 2.978 (0.27) & 2.977 (0.27) & 2.990 (0.41) & 3.003 (0.26) & 3.033 (0.28) & 3.114 (0.48)\vspace{1mm} \\
\hline
\end{tabular}
\caption{Mean and SD of 10,000 simulated estimates of $\mu^{(1)}$. The covariates $X$ were generated from Gaussian distributions, and the outcome regression was obtained by the Gaussian MLE using non-outliers. The characters "T" and "F" denote the correct and the incorrect modeling, respectively.}
\end{table}
\end{landscape}

% latex table generated in R 4.0.3 by xtable 1.8-4 package
% Thu Jun 03 09:46:16 2021
\begin{landscape}
  \begin{table}[ht]
  \centering
  \scriptsize
  \begin{tabular}{llrrrrrrr}
    \hline
  &  & {\bf No contam.} & \multicolumn{3}{c}{\bf Homogeneous} & \multicolumn{3}{c}{\bf Heterogeneous} \\ 
   &  &  & $\varepsilon=0.05$ & $\varepsilon=0.10$ & $\varepsilon=0.20$ & $\varepsilon=0.05$ & $\varepsilon=0.10$ & $\varepsilon=0.20$ \\
    \hline
IPW(T/-) & Naive & 0.009 & 0.010 & 0.012 & 0.009 & 0.010 & 0.010 & 0.012 \\ 
   & median (Firpo) & 2.019 & 2.035 & 2.052 & 2.040 & 2.049 & 2.049 & 2.044 \\ 
   & median (Zhang-IPW) & 0.153 & 0.166 & 0.165 & 0.162 & 0.194 & 0.195 & 0.203\vspace{1mm} \\ 
   & DP-IPW ($\gamma=0.1$) & 33.713 & 44.646 & 62.069 & 80.462 & 44.256 & 60.934 & 78.925 \\ 
   & DP-IPW ($\gamma=0.5$) & 48.822 & 50.308 & 48.108 & 46.894 & 49.748 & 48.500 & 47.246 \\ 
   & DP-IPW ($\gamma=1.0$) & 65.508 & 64.401 & 62.285 & 55.019 & 63.856 & 61.479 & 54.070\vspace{1mm} \\
\hline
\end{tabular}
\caption{Mean computation time (ms) of 10,000 simulations. The covariates $X$ were generated from Gaussian distributions, and the outcome regression was obtained by the Gaussian MLE using non-outliers. The characters "T" and "F" denote the correct and the incorrect modeling, respectively.}
\end{table}
\end{landscape}

% latex table generated in R 4.0.3 by xtable 1.8-4 package
% Thu Jun 03 09:46:16 2021
\begin{landscape}
\begin{table}[ht]
\centering
\scriptsize
\begin{tabular}{llccccccc}
  \hline
  \multicolumn{2}{l}{} & {\bf No contam.} & \multicolumn{3}{c}{\bf Homogeneous} & \multicolumn{3}{c}{\bf Heterogeneous} \\ 
  &  &  & $\varepsilon=0.05$ & $\varepsilon=0.10$ & $\varepsilon=0.20$ & $\varepsilon=0.05$ & $\varepsilon=0.10$ & $\varepsilon=0.20$ \\ 
    \hline
IPW(T/-) & Naive & 0.222 & 0.957 & 1.683 & 3.153 & 0.993 & 1.752 & 3.253 \\ 
   & median (Firpo) & 0.257 & 0.294 & 0.367 & 0.649 & 0.306 & 0.409 & 0.769 \\ 
   & median (Zhang-IPW) & 0.257 & 0.294 & 0.367 & 0.649 & 0.306 & 0.409 & 0.769\vspace{1mm} \\ 
   & DP-IPW ($\gamma=0.1$) & 0.218 & 0.276 & 0.531 & 2.263 & 0.293 & 0.609 & 2.377 \\ 
   & DP-IPW ($\gamma=0.5$) & 0.227 & 0.249 & 0.272 & 0.639 & 0.245 & 0.287 & 0.726 \\ 
   & DP-IPW ($\gamma=1.0$) & 0.261 & 0.271 & 0.275 & 0.413 & 0.262 & 0.281 & 0.498\vspace{1mm} \\
\hline
\end{tabular}
\caption{Results of the comparative study. Each figure is RMSE between each method and the true value. The covariates $X$ were generated from Gaussian distributions, and the outcome regression was obtained by the unnormalized Gaussian modeling. The characters "T" and "F" denote the correct and the incorrect modeling, respectively.}
\end{table}
\end{landscape}

% latex table generated in R 4.0.3 by xtable 1.8-4 package
% Thu Jun 03 09:46:16 2021
\begin{landscape}
  \begin{table}[ht]
  \centering
  \scriptsize
  \begin{tabular}{llccccccc}
    \hline
  &  & {\bf No contam.} & \multicolumn{3}{c}{\bf Homogeneous} & \multicolumn{3}{c}{\bf Heterogeneous} \\ 
   &  &  & $\varepsilon=0.05$ & $\varepsilon=0.10$ & $\varepsilon=0.20$ & $\varepsilon=0.05$ & $\varepsilon=0.10$ & $\varepsilon=0.20$ \\
  \hline
IPW(T/-) & Naive & 3.004 (0.22) & 3.749 (0.59) & 4.493 (0.78) & 5.983 (1.02) & 3.766 (0.63) & 4.536 (0.84) & 6.070 (1.08) \\ 
   & median (Firpo) & 2.990 (0.26) & 3.091 (0.28) & 3.205 (0.30) & 3.490 (0.43) & 3.116 (0.28) & 3.259 (0.32) & 3.605 (0.47) \\ 
   & median (Zhang-IPW) & 2.990 (0.26) & 3.091 (0.28) & 3.205 (0.30) & 3.490 (0.43) & 3.116 (0.28) & 3.259 (0.32) & 3.605 (0.47)\vspace{1mm} \\ 
   & DP-IPW ($\gamma=0.1$) & 2.998 (0.22) & 3.030 (0.27) & 3.142 (0.51) & 4.492 (1.70) & 3.056 (0.29) & 3.209 (0.57) & 4.620 (1.74) \\ 
   & DP-IPW ($\gamma=0.5$) & 2.986 (0.23) & 2.987 (0.25) & 2.989 (0.27) & 3.052 (0.64) & 3.011 (0.24) & 3.042 (0.28) & 3.173 (0.70) \\ 
   & DP-IPW ($\gamma=1.0$) & 2.980 (0.26) & 2.978 (0.27) & 2.977 (0.27) & 2.990 (0.41) & 3.003 (0.26) & 3.033 (0.28) & 3.114 (0.48)\vspace{1mm} \\
\hline
\end{tabular}
\caption{Mean and SD of 10,000 simulated estimates of $\mu^{(1)}$. The covariates $X$ were generated from Gaussian distributions, and the outcome regression was obtained by the unnormalized Gaussian modeling. The characters "T" and "F" denote the correct and the incorrect modeling, respectively.}
\end{table}
\end{landscape}

% latex table generated in R 4.0.3 by xtable 1.8-4 package
% Thu Jun 03 09:46:16 2021
\begin{landscape}
  \begin{table}[ht]
  \centering
  \scriptsize
  \begin{tabular}{llrrrrrrr}
    \hline
  &  & {\bf No contam.} & \multicolumn{3}{c}{\bf Homogeneous} & \multicolumn{3}{c}{\bf Heterogeneous} \\ 
   &  &  & $\varepsilon=0.05$ & $\varepsilon=0.10$ & $\varepsilon=0.20$ & $\varepsilon=0.05$ & $\varepsilon=0.10$ & $\varepsilon=0.20$ \\
    \hline
IPW(T/-) & Naive & 0.008 & 0.011 & 0.009 & 0.010 & 0.009 & 0.008 & 0.006 \\ 
   & median (Firpo) & 2.028 & 2.038 & 2.009 & 1.992 & 1.591 & 1.604 & 1.610 \\ 
   & median (Zhang-IPW) & 0.164 & 0.184 & 0.159 & 0.167 & 0.172 & 0.166 & 0.168\vspace{1mm} \\ 
   & DP-IPW ($\gamma=0.1$) & 33.735 & 44.282 & 60.503 & 79.176 & 39.488 & 54.135 & 70.099 \\ 
   & DP-IPW ($\gamma=0.5$) & 48.650 & 49.775 & 47.834 & 46.162 & 43.983 & 43.577 & 41.231 \\ 
   & DP-IPW ($\gamma=1.0$) & 65.725 & 64.715 & 61.072 & 53.993 & 56.336 & 54.092 & 48.227\vspace{1mm} \\
\hline
\end{tabular}
\caption{Mean computation time (ms) of 10,000 simulations. The covariates $X$ were generated from Gaussian distributions, and the outcome regression was obtained by the unnormalized Gaussian modeling. The characters "T" and "F" denote the correct and the incorrect modeling, respectively.}
\end{table}
\end{landscape}

% latex table generated in R 4.0.3 by xtable 1.8-4 package
% Thu Jun 03 09:46:16 2021
\begin{landscape}
\begin{table}[ht]
\centering
\scriptsize
\begin{tabular}{llccccccc}
  \hline
  \multicolumn{2}{l}{} & {\bf No contam.} & \multicolumn{3}{c}{\bf Homogeneous} & \multicolumn{3}{c}{\bf Heterogeneous} \\ 
  &  &  & $\varepsilon=0.05$ & $\varepsilon=0.10$ & $\varepsilon=0.20$ & $\varepsilon=0.05$ & $\varepsilon=0.10$ & $\varepsilon=0.20$ \\ 
    \hline
IPW(T/-) & Naive & 0.197 & 0.952 & 1.711 & 3.174 & 1.000 & 1.762 & 3.269 \\ 
   & median (Firpo) & 0.278 & 0.319 & 0.407 & 0.705 & 0.334 & 0.447 & 0.785 \\ 
   & median (Zhang-IPW) & 0.278 & 0.319 & 0.407 & 0.705 & 0.334 & 0.447 & 0.785\vspace{1mm} \\ 
   & DP-IPW ($\gamma=0.1$) & 0.199 & 0.230 & 0.567 & 2.329 & 0.253 & 0.609 & 2.429 \\ 
   & DP-IPW ($\gamma=0.5$) & 0.223 & 0.230 & 0.239 & 0.703 & 0.230 & 0.247 & 0.690 \\ 
   & DP-IPW ($\gamma=1.0$) & 0.273 & 0.273 & 0.273 & 0.432 & 0.272 & 0.279 & 0.394\vspace{1mm} \\
\hline
\end{tabular}
\caption{Results of the comparative study. Each figure is RMSE between each method and the true value. The covariates $X$ were generated from uniform distributions, and the outcome regression was obtained by the Gaussian MLE using non-outliers. The characters "T" and "F" denote the correct and the incorrect modeling, respectively.}
\end{table}
\end{landscape}

% latex table generated in R 4.0.3 by xtable 1.8-4 package
% Thu Jun 03 09:46:16 2021
\begin{landscape}
  \begin{table}[ht]
  \centering
  \scriptsize
  \begin{tabular}{llccccccc}
    \hline
  &  & {\bf No contam.} & \multicolumn{3}{c}{\bf Homogeneous} & \multicolumn{3}{c}{\bf Heterogeneous} \\ 
   &  &  & $\varepsilon=0.05$ & $\varepsilon=0.10$ & $\varepsilon=0.20$ & $\varepsilon=0.05$ & $\varepsilon=0.10$ & $\varepsilon=0.20$ \\
  \hline
IPW(T/-) & Naive & 3.003 (0.20) & 3.759 (0.57) & 4.521 (0.79) & 6.004 (1.02) & 3.783 (0.62) & 4.556 (0.83) & 6.090 (1.07) \\ 
   & median (Firpo) & 2.985 (0.28) & 3.099 (0.30) & 3.233 (0.33) & 3.541 (0.45) & 3.130 (0.31) & 3.292 (0.34) & 3.666 (0.42) \\ 
   & median (Zhang-IPW) & 2.985 (0.28) & 3.099 (0.30) & 3.233 (0.33) & 3.541 (0.45) & 3.130 (0.31) & 3.292 (0.34) & 3.666 (0.42)\vspace{1mm} \\ 
   & DP-IPW ($\gamma=0.1$) & 3.000 (0.20) & 3.038 (0.23) & 3.186 (0.54) & 4.590 (1.70) & 3.066 (0.24) & 3.241 (0.56) & 4.713 (1.72) \\ 
   & DP-IPW ($\gamma=0.5$) & 2.991 (0.22) & 2.990 (0.23) & 2.994 (0.24) & 3.075 (0.70) & 3.017 (0.23) & 3.048 (0.24) & 3.185 (0.66) \\ 
   & DP-IPW ($\gamma=1.0$) & 2.979 (0.27) & 2.978 (0.27) & 2.982 (0.27) & 2.998 (0.43) & 3.009 (0.27) & 3.043 (0.28) & 3.123 (0.37)\vspace{1mm} \\
\hline
\end{tabular}
\caption{Mean and SD of 10,000 simulated estimates of $\mu^{(1)}$. The covariates $X$ were generated from uniform distributions, and the outcome regression was obtained by the Gaussian MLE using non-outliers. The characters "T" and "F" denote the correct and the incorrect modeling, respectively.}
\end{table}
\end{landscape}

% latex table generated in R 4.0.3 by xtable 1.8-4 package
% Thu Jun 03 09:46:16 2021
\begin{landscape}
  \begin{table}[ht]
  \centering
  \scriptsize
  \begin{tabular}{llrrrrrrr}
    \hline
  &  & {\bf No contam.} & \multicolumn{3}{c}{\bf Homogeneous} & \multicolumn{3}{c}{\bf Heterogeneous} \\ 
   &  &  & $\varepsilon=0.05$ & $\varepsilon=0.10$ & $\varepsilon=0.20$ & $\varepsilon=0.05$ & $\varepsilon=0.10$ & $\varepsilon=0.20$ \\
    \hline
IPW(T/-) & Naive & 0.006 & 0.008 & 0.008 & 0.010 & 0.009 & 0.008 & 0.011 \\ 
   & median (Firpo) & 1.624 & 1.554 & 1.598 & 1.596 & 1.640 & 1.599 & 1.595 \\ 
   & median (Zhang-IPW) & 0.144 & 0.155 & 0.146 & 0.149 & 0.163 & 0.166 & 0.166\vspace{1mm} \\ 
   & DP-IPW ($\gamma=0.1$) & 28.899 & 40.705 & 57.891 & 72.877 & 41.438 & 58.884 & 72.039 \\ 
   & DP-IPW ($\gamma=0.5$) & 41.124 & 43.104 & 41.639 & 40.541 & 43.319 & 42.599 & 40.527 \\ 
   & DP-IPW ($\gamma=1.0$) & 58.435 & 57.102 & 56.009 & 48.982 & 57.823 & 54.624 & 48.146\vspace{1mm} \\
\hline
\end{tabular}
\caption{Mean computation time (ms) of 10,000 simulations. The covariates $X$ were generated from uniform distributions, and the outcome regression was obtained by the Gaussian MLE over non-outliers. The characters "T" and "F" denote the correct and the incorrect modeling, respectively.}
\end{table}
\end{landscape}

% latex table generated in R 4.0.3 by xtable 1.8-4 package
% Thu Jun 03 09:46:17 2021
\begin{landscape}
\begin{table}[ht]
\centering
\scriptsize
\begin{tabular}{llccccccc}
  \hline
  \multicolumn{2}{l}{} & {\bf No contam.} & \multicolumn{3}{c}{\bf Homogeneous} & \multicolumn{3}{c}{\bf Heterogeneous} \\ 
  &  &  & $\varepsilon=0.05$ & $\varepsilon=0.10$ & $\varepsilon=0.20$ & $\varepsilon=0.05$ & $\varepsilon=0.10$ & $\varepsilon=0.20$ \\ 
    \hline
IPW(T/-) & Naive & 0.197 & 0.952 & 1.711 & 3.174 & 1.000 & 1.762 & 3.269 \\ 
   & median (Firpo) & 0.278 & 0.319 & 0.407 & 0.705 & 0.334 & 0.447 & 0.785 \\ 
   & median (Zhang-IPW) & 0.278 & 0.319 & 0.407 & 0.705 & 0.334 & 0.447 & 0.785\vspace{1mm} \\ 
   & DP-IPW ($\gamma=0.1$) & 0.199 & 0.230 & 0.567 & 2.329 & 0.253 & 0.609 & 2.429 \\ 
   & DP-IPW ($\gamma=0.5$) & 0.223 & 0.230 & 0.239 & 0.703 & 0.230 & 0.247 & 0.690 \\ 
   & DP-IPW ($\gamma=1.0$) & 0.273 & 0.273 & 0.273 & 0.432 & 0.272 & 0.279 & 0.394\vspace{1mm} \\
\hline
\end{tabular}
\caption{Results of the comparative study. Each figure is RMSE between each method and the true value. The covariates $X$ were generated from uniform distributions, and the outcome regression was obtained by the unnormalized Gaussian modeling. The characters "T" and "F" denote the correct and the incorrect modeling, respectively.}
\end{table}
\end{landscape}

% latex table generated in R 4.0.3 by xtable 1.8-4 package
% Thu Jun 03 09:46:17 2021
\begin{landscape}
  \begin{table}[ht]
  \centering
  \scriptsize
  \begin{tabular}{llccccccc}
    \hline
  &  & {\bf No contam.} & \multicolumn{3}{c}{\bf Homogeneous} & \multicolumn{3}{c}{\bf Heterogeneous} \\ 
   &  &  & $\varepsilon=0.05$ & $\varepsilon=0.10$ & $\varepsilon=0.20$ & $\varepsilon=0.05$ & $\varepsilon=0.10$ & $\varepsilon=0.20$ \\
  \hline
IPW(T/-) & Naive & 3.003 (0.20) & 3.759 (0.57) & 4.521 (0.79) & 6.004 (1.02) & 3.783 (0.62) & 4.556 (0.83) & 6.090 (1.07) \\ 
   & median (Firpo) & 2.985 (0.28) & 3.099 (0.30) & 3.233 (0.33) & 3.541 (0.45) & 3.130 (0.31) & 3.292 (0.34) & 3.666 (0.42) \\ 
   & median (Zhang-IPW) & 2.985 (0.28) & 3.099 (0.30) & 3.233 (0.33) & 3.541 (0.45) & 3.130 (0.31) & 3.292 (0.34) & 3.666 (0.42)\vspace{1mm} \\ 
   & DP-IPW ($\gamma=0.1$) & 3.000 (0.20) & 3.038 (0.23) & 3.186 (0.54) & 4.590 (1.70) & 3.066 (0.24) & 3.241 (0.56) & 4.713 (1.72) \\ 
   & DP-IPW ($\gamma=0.5$) & 2.991 (0.22) & 2.990 (0.23) & 2.994 (0.24) & 3.075 (0.70) & 3.017 (0.23) & 3.048 (0.24) & 3.185 (0.66) \\ 
   & DP-IPW ($\gamma=1.0$) & 2.979 (0.27) & 2.978 (0.27) & 2.982 (0.27) & 2.998 (0.43) & 3.009 (0.27) & 3.043 (0.28) & 3.123 (0.37)\vspace{1mm} \\
\hline
\end{tabular}
\caption{Mean and SD of 10,000 simulated estimates of $\mu^{(1)}$. The covariates $X$ were generated from uniform distributions, and the outcome regression was obtained by the unnormalized Gaussian modeling. The characters "T" and "F" denote the correct and the incorrect modeling, respectively.}
\end{table}
\end{landscape}

% latex table generated in R 4.0.3 by xtable 1.8-4 package
% Thu Jun 03 09:46:17 2021
\begin{landscape}
  \begin{table}[ht]
  \centering
  \scriptsize
  \begin{tabular}{llrrrrrrr}
    \hline
  &  & {\bf No contam.} & \multicolumn{3}{c}{\bf Homogeneous} & \multicolumn{3}{c}{\bf Heterogeneous} \\ 
   &  &  & $\varepsilon=0.05$ & $\varepsilon=0.10$ & $\varepsilon=0.20$ & $\varepsilon=0.05$ & $\varepsilon=0.10$ & $\varepsilon=0.20$ \\
    \hline
IPW(T/-) & Naive & 0.012 & 0.009 & 0.011 & 0.008 & 0.009 & 0.008 & 0.008 \\ 
   & median (Firpo) & 2.061 & 2.018 & 2.001 & 2.014 & 1.593 & 1.588 & 1.591 \\ 
   & median (Zhang-IPW) & 0.159 & 0.158 & 0.155 & 0.174 & 0.160 & 0.162 & 0.162\vspace{1mm} \\ 
   & DP-IPW ($\gamma=0.1$) & 32.502 & 46.715 & 64.496 & 80.991 & 41.595 & 58.327 & 72.898 \\ 
   & DP-IPW ($\gamma=0.5$) & 45.786 & 48.372 & 46.799 & 45.370 & 43.191 & 41.859 & 40.298 \\ 
   & DP-IPW ($\gamma=1.0$) & 67.400 & 64.207 & 61.061 & 54.756 & 57.049 & 54.186 & 47.922\vspace{1mm} \\
\hline
\end{tabular}
\caption{Mean computation time (ms) of 10,000 simulations. The covariates $X$ were generated from uniform distributions, and the outcome regression was obtained by the unnormalized Gaussian modeling. The characters "T" and "F" denote the correct and the incorrect modeling, respectively.}
\end{table}
\end{landscape}

\end{document}